\documentclass[11pt]{article}

\usepackage{amsmath}
\usepackage{amsfonts}
\usepackage{amssymb}
\usepackage{amsthm}
\usepackage{multirow}
\usepackage{graphicx}
\usepackage{cite}
\usepackage{empheq}
\usepackage{caption}
\usepackage{subcaption}

\newtheorem{lemma}{Lemma}
\newtheorem{theorem}{Theorem}

\newtheorem{definition}{Definition}
\newtheorem{remark}{Remark}

\title{On the quasi-unconditional stability of BDF-ADI solvers for the
  compressible Navier-Stokes equations}

\author{Oscar P. Bruno and Max Cubillos}

\begin{document}

\date{}

\maketitle

\begin{abstract}
  The companion paper ``Higher-order in time quasi-unconditionally
  stable ADI solvers for the compressible Navier-Stokes equations in
  2D and 3D curvilinear domains'', which is referred to as Part~I in
  what follows, introduces ADI (Alternating Direction Implicit)
  solvers of higher orders of temporal accuracy (orders $s = 2$ to 6)
  for the compressible Navier-Stokes equations in two- and
  three-dimensional space. The proposed methodology employs the
  backward differentiation formulae (BDF) together with a
  quasilinear-like formulation, high-order extrapolation for nonlinear
  components, and the Douglas-Gunn splitting. A variety of numerical
  results presented in Part~I demonstrate in practice the theoretical
  convergence rates enjoyed by these algorithms, as well as their
  excellent accuracy and stability properties for a wide range of
  Reynolds numbers. In particular, the proposed schemes enjoy a
  certain property of ``quasi-unconditional stability'': for small
  enough (problem-dependent) fixed values of the time-step $\Delta t$,
  these algorithms are stable for arbitrarily fine spatial
  discretizations. The present contribution presents a mathematical
  basis for the performance of these algorithms. Short of providing
  stability theorems for the full BDF-ADI Navier-Stokes solvers, this
  paper puts forth proofs of unconditional stability and
  quasi-unconditional stability for BDF-ADI schemes as well as some
  related un-split BDF schemes, for a variety of related linear model
  problems in one, two and three spatial dimensions, and for schemes
  of orders $2\leq s\leq 6$ of temporal accuracy. Additionally, a set
  of numerical tests presented in this paper for the compressible
  Navier-Stokes equation indicate that quasi-unconditional stability
  carries over to the fully non-linear context.
\end{abstract}

\newpage

\section{Introduction} \label{sec:Introduction}

The companion paper~\cite{bruno_higher-order_2015}, which is referred
to as Part~I in what follows, introduces ADI (Alternating Direction
Implicit) solvers of higher orders of time-accuracy (orders $s = 2$ to
6) for the compressible Navier-Stokes equations in two- and
three-dimensional curvilinear domains. Implicit solvers, even of ADI
type, are generally more expensive per time-step, for a given spatial
discretization, than explicit solvers, but use of efficient implicit
solvers can be advantageous whenever the time-step restrictions
imposed by the mesh spacing $h$ are too severe. The proposed
methodology employs the BDF (backward differentiation formulae)
multi-step ODE solvers (which are known for their robust stability
properties) together with a quasilinear-like formulation and
high-order extrapolation for nonlinear components (which gives rise to
a linear problem that can be solved efficiently by means of standard
linear algebra solvers) and the Douglas-Gunn splitting (an ADI
strategy that greatly simplifies the treatment of boundary conditions
while retaining the order of time-accuracy of the solver).

As discussed in Part~I, the proposed BDF-ADI solvers are the first
ADI-based Navier-Stokes solvers for which high-order time-accuracy has
been demonstrated. In spite of the nominal second order of
time-accuracy inherent in the celebrated Beam and Warming
method~\cite{beam_implicit_1978}
(cf. also~\cite{beam_alternating_1980,warming_extension_1979}),
previous ADI solvers for the Navier-Stokes equations have not
demonstrated time-convergence of orders higher than one under general
non-periodic physical boundary conditions. Part~I demonstrates the
properties of the proposed schemes by means of a variety of numerical
experiments; the present paper, in turn, provides a theoretical basis
for the observed algorithmic stability traits. Restricting attention
to linear model problems related to the fully nonlinear Navier-Stokes
equations---including the advection-diffusion equation in one, two
and three spatial dimensions as well as general parabolic and
hyperbolic linear equations in two spatial dimensions---the present
contribution provides proofs of unconditional stability for BDF-ADI
schemes of second order as well as proofs of quasi-unconditional
stability (Definition~\ref{def:QuasiUnconditional} below) for certain
related BDF-based unsplit (non-ADI) schemes of orders $3\leq s\leq 6$.
Further, a variety of numerical tests presented in
Section~\ref{sec:BDFLinearNS} indicate that the property of
quasi-unconditional stability carries over to the BDF-ADI solvers for
the fully non-linear Navier-Stokes equations.

The BDF-ADI methodology mentioned above can be applied in conjunction
with a variety of spatial discretizations. For definiteness, in this
contribution attention is restricted to Chebyshev, Legendre and
Fourier spectral spatial approximations. The resulting one-dimensional
boundary value problems arising from these discretizations involve
full matrices which generally cannot be inverted efficiently by means
of a direct solver. However, by relying on fast transforms these
systems can be solved effectively on the basis of the GMRES iterative
solver; details in these regards are presented in
Part~I. Additionally, as detailed in that reference, in order to
ensure stability for the fully nonlinear Navier-Stokes equations a
mild spectral filter is used.

Perhaps the existence of Dahlquist's second barrier may explain the
widespread use of implicit methods of orders less than or equal to two
in the present context (such as backward Euler, the trapezoidal rule
and BDF2, all of which are A-stable), and the virtual absence of
implicit methods of orders higher than two---despite the widespread
use of the fourth order Runge-Kutta and Adams-Bashforth explicit
counterparts.  Clearly, A-stability is not necessary for all
problems---for example, any method whose stability region contains the
negative real axis (such as the BDF methods of orders two to six)
generally results in an unconditionally stable solver for the heat
equation. A number of important questions thus arise: Do the
compressible Navier-Stokes equations inherently require A-stability?
Are the stability constraints of all higher-order implicit methods too
stringent to be useful in the Navier-Stokes context?  How close to
unconditionally stable can a Navier-Stokes solver be whose temporal
order of accuracy is higher than two?

Clear answers to these questions are not available in the extant
literature; the present work seeks to provide a theoretical
understanding in these regards.  To illustrate the present state of
the art concerning such matters we mention the 2002
reference~\cite{bijl_implicit_2002}, which compares various implicit
methods for the Navier-Stokes equations, where we read: ``Practical
experience indicates that large-scale engineering computations are
seldom stable if run with BDF4.  The BDF3 scheme, with its smaller
regions of instability, is often stable but diverges for certain
problems and some spatial operators.  Thus, a reasonable practitioner
might use the BDF2 scheme exclusively for large-scale computations.''
It must be noted, however, that neither the
article~\cite{bijl_implicit_2002} nor the references it cites
investigate in detail the stability restrictions associated with the
BDF methods order $s\geq 2$, either theoretically or
experimentally. But higher-order methods can be useful: as demonstrated
in Part~I, methods of order higher than two give rise to very
significant advantages for certain classes of problems---especially
for large-scale computations for which the temporal dispersion
inherent in low-order approaches would make it necessary to use
inordinately small time-steps.

The recent 2015 article~\cite{forti_semiimplicit_2015}, in turn,
presents applications of the BDF scheme up to third order of time
accuracy in a finite element context for the incompressible
Navier-Stokes equations with turbulence modelling. This contribution
does not discuss stability restrictions for the third order solver,
and, in fact, it only presents numerical examples resulting from use
of BDF1 and BDF2.  The 2010 contribution~\cite{karaa_hybrid_2010},
which considers a three-dimensional advection-diffusion equation,
presents various ADI-type schemes, one of which is based on BDF3.  The
BDF3 stability analysis in that paper, however, is restricted to the
purely diffusive case.

This paper is organized as follows: Section~\ref{sec:BDF-ADI} presents
a brief derivation of the BDF-ADI method for the two-dimensional
pressure-free momentum equation. (A derivation for the full
Navier-Stokes equations is given in Part~I, but the specialized
derivation presented here may prove valuable in view of its relative
simplicity.) Section~\ref{sec:StabilityDiscussion} then briefly
reviews relevant notions from classical stability theory as well as
the concept of quasi-unconditional stability introduced in
Part~I. Section~\ref{sec:BDF2Stability} presents proofs of (classical)
unconditional stability for two-dimensional BDF-ADI schemes of order
$s=2$ specialized to the linear constant coefficient periodic
advection equation as well as the linear constant coefficient periodic
and non-periodic parabolic equations. Proofs of quasi-unconditional
stability for the (non-ADI) BDF methods of orders $s=2,\dots,6$
applied to the constant coefficient advection-diffusion equation in
one, two, and three spatial dimensions are presented in
Section~\ref{sec:HighOrderBDFStability}, along with comparisons of the
stability constraints arising from these BDF solvers and the
commonly-used explicit Adams-Bashforth solvers of orders three and
four. Section~\ref{sec:BDFLinearNS} provides numerical tests that
indicate that the BDF-ADI methods of orders $s=2$ to $6$ for the the
full three-dimensional Navier-Stokes equations enjoy the property of
quasi-unconditional stability; a wide variety of additional numerical
experiments are presented in Part~I. Section~\ref{sec:Conclusions},
finally, presents a few concluding remarks.

\section{The BDF-ADI scheme} \label{sec:BDF-ADI}

In this section we present a derivation of the BDF-ADI scheme in a
somewhat simplified context, restricting attention to the
two-dimensional pressure-free momentum equation
\begin{equation} \label{eq:PressureFreeMomentum}
  \mathbf{u}_{t}+\mathbf{u}\cdot\nabla\mathbf{u} =  \mu(\,\Delta
  \mathbf u + \frac 13 \nabla(\nabla\cdot\mathbf u)\,)
\end{equation}
in Cartesian coordinates for the velocity vector $\mathbf{u} =
(u,v)^\mathrm{T}$. The present derivation may thus be more readily
accessible than the one presented in Part~I for the full Navier-Stokes
equations under curvilinear coordinates.  Like the BDF-ADI
Navier-Stokes algorithms presented in Part~I, the schemes discussed in
this section incorporate three main elements, namely, 1)~A BDF-based
time discretization; 2) High-order extrapolation of relevant factors
in quasilinear terms (the full compressible Navier-Stokes solver
presented in Part~I utilizes a similar procedure for non-quasilinear
terms); and 3) The Douglas-Gunn ADI splitting.

The semi-discrete BDF scheme of order $s$ for
equation~\eqref{eq:PressureFreeMomentum} is given by
\begin{equation} \label{eq:MomentumSemiDiscrete}
  \mathbf u^{n+1} = \sum_{j=0}^{s-1} a_j \mathbf u^{n-j} + b\Delta
  t\left( -\mathbf{u}^{n+1} \cdot\nabla\mathbf{u}^{n+1} + \mu(\,\Delta
    \mathbf u^{n+1} + \frac 13 \nabla(\nabla\cdot\mathbf
    u^{n+1})\,)\right)
\end{equation}
where $a_j$ and $b$ are the order-$s$ BDF coefficients (see
e.g.~\cite[Ch. 3.12]{lambert_numerical_1991} or~Part~I); the
truncation error associated with this scheme is a quantity of order
$\mathcal O((\Delta t)^{s+1})$. This equation is quasi-linear: the
derivatives of the solution appear linearly in the equation. Of
course, the full compressible Navier-Stokes equations contain several
non-quasilinear nonlinear terms. As detailed in Part~I, by introducing
a certain ``quasilinear-like'' form of the equations, all such
nonlinear terms can be treated by an approach similar to the one
described in this section. In preparation for a forthcoming ADI
splitting we consider the somewhat more detailed form
\begin{equation} \label{eq:MomentumNonlinear_1}
\begin{split}
\mathbf u^{n+1}
=  \sum_{j=0}^{s-1} a_j \mathbf u^{n-j} 
+  b\Delta t \bigg( -u^{n+1} \partial_x
- v^{n+1} &\partial_y 
+ \mu \left( \begin{smallmatrix}
4/3 & 0 \\ 0 & 1
\end{smallmatrix} \right)   \partial_x^2
\\ & + \mu \left( \begin{smallmatrix}
1 & 0 \\ 0 & 4/3
\end{smallmatrix} \right) \partial_y^2
+ \mu \left( \begin{smallmatrix}
0 & 1/3 \\ 1/3 & 0
\end{smallmatrix} \right) \partial_x\partial_y \bigg) \mathbf u^{n+1}
\end{split}
\end{equation}
of equation~\eqref{eq:MomentumSemiDiscrete}, which we then rewrite as 
\begin{equation} \label{eq:MomentumNonlinear}
\begin{split}
  \Bigg( I + b\Delta t \bigg( u^{n+1} \partial_x +  
    v^{n+1} \partial_y + \mu \left( \begin{smallmatrix} 4/3 & 0 \\ 0 &
        1
\end{smallmatrix} \right) & \partial_x^2
+ \mu \left( \begin{smallmatrix}
1 & 0 \\ 0 & 4/3
\end{smallmatrix} \right) \partial_y^2 \bigg) \Bigg)
\mathbf u^{n+1} \\
&= \sum_{j=0}^{s-1} a_j \mathbf u^{n-j} 
+ b\Delta t \mu \left( \begin{smallmatrix}
0 & 1/3 \\ 1/3 & 0
\end{smallmatrix} \right) \partial_x\partial_y  \mathbf u^{n+1} .
\end{split}
\end{equation}

Upon spatial discretization, the solution of
equation~\eqref{eq:MomentumNonlinear} for the unknown velocity field
$\mathbf u^{n+1}$ amounts to inversion of a (generally large)
non-linear system of equations. In order to avoid inversion of such
nonlinear systems we rely on high-order extrapolation of certain
non-differentiated terms. This procedure eliminates the nonlinearities
present in the equation while preserving the order of temporal
accuracy of the algorithm. In detail, let $P_u$ (resp. $P_v$) denote
the polynomial of degree $s-1$ that passes through
$(t^{n-j+1},u^{n-j+1})$ (resp. through $(t^{n-j+1},v^{n-j+1})$) for
$1\leq j\leq s$, and define $\widetilde u_s^{n+1} = P_u(t^{n+1})$
(resp. $\widetilde v_s^{n+1} = P_v(t^{n+1}$)) and $\mathbf{\widetilde
  u}_s^{n+1} = (\widetilde u_s^{n+1},\widetilde v_s^{n+1}) $.  Then,
substituting $\widetilde u_s^{n+1}$ and $\widetilde v_s^{n+1}$
(resp. $\mathbf{\widetilde u}_s^{n+1}$) for the un-differentiated
terms $u^{n+1}$ and $v^{n+1}$ (resp. for the mixed derivative term) in
equation~\eqref{eq:MomentumNonlinear}, the alternative
variable-coefficient {\em linear} semi-discrete scheme
\begin{align} \label{eq:MomentumExtrapolated}
& \bigg[ I + b\Delta t \bigg( \widetilde u_s^{n+1} \partial_x
+ \widetilde v_s^{n+1} \partial_y
+ \mu \left( \begin{smallmatrix}
4/3 & 0 \\ 0 & 1
\end{smallmatrix} \right) \partial_x^2
+ \mu \left( \begin{smallmatrix}
1 & 0 \\ 0 & 4/3
\end{smallmatrix} \right) \partial_y^2 \bigg) \bigg]
\mathbf u^{n+1} \nonumber \\
&= \sum_{j=0}^{s-1} a_j \mathbf u^{n-j} 
+ b\Delta t \mu \left( \begin{smallmatrix}
0 & 1/3 \\ 1/3 & 0
\end{smallmatrix} \right) \partial_x\partial_y \mathbf{\widetilde u}_s^{n+1} 
\end{align}
results. Clearly the truncation errors inherent in the linear
scheme~\eqref{eq:MomentumExtrapolated} are of the same order as those associated
with the original nonlinear scheme~\eqref{eq:MomentumNonlinear}.

Even though equation~\eqref{eq:MomentumExtrapolated} is linear,
solution of (a spatially discretized version of) this equation
requires inversion of a generally exceedingly large linear system at
each time step. To avoid this difficulty we resort to a strategy of
ADI type~\cite{peaceman_numerical_1955} and, more explicitly, to the
Douglas-Gunn splitting~\cite{douglas_jr._general_1964}.  To derive the
Douglas-Gunn splitting we re-express
equation~\eqref{eq:MomentumExtrapolated} in the factored form
\begin{align} \label{eq:MomentumFactored}
& \bigg[ I + b\Delta t \bigg( \widetilde u_s^{n+1}  \partial_x 
+ \mu \left( \begin{smallmatrix}
4/3 & 0 \\ 0 & 1
\end{smallmatrix} \right)  \partial_x^2  \bigg) \bigg]
\bigg[ I + b\Delta t \bigg( \widetilde v_s^{n+1}   \partial_y 
+ \mu \left( \begin{smallmatrix}
1 & 0 \\ 0 & 4/3
\end{smallmatrix} \right)  \partial_y^2  \bigg) \bigg]
\mathbf u^{n+1} \nonumber \\ 
& \quad = \sum_{j=0}^{s-1} a_j \mathbf u^{n-j} 
+ b\Delta t \mu \left( \begin{smallmatrix}
0 & 1/3 \\ 1/3 & 0
\end{smallmatrix} \right) \partial_x\partial_y \mathbf{\widetilde u}_s^{n+1} \nonumber \\
& \quad \ + (b\Delta t)^2  \bigg( \widetilde u_s^{n+1}   \partial_x 
+ \mu \left( \begin{smallmatrix}
4/3 & 0 \\ 0 & 1
\end{smallmatrix} \right) \partial_x^2  \bigg)
\bigg( \widetilde v_s^{n+1} \partial_y 
+ \mu \left( \begin{smallmatrix}
1 & 0 \\ 0 & 4/3
\end{smallmatrix} \right) \partial_y^2  \bigg) \mathbf{\widetilde u}_{s-1}^{n+1}.
\end{align}
We specially mention the presence of terms on the right hand side this
equation which only depend on solution values at times
$t^n,\dots,t^{n-s+1}$, and which have been incorporated to obtain an
equation that is equivalent to~\eqref{eq:MomentumExtrapolated} up to
order $\mathcal O((\Delta t)^{s+1})$.

\begin{remark}
  It is important to note that, although $\mathbf{\widetilde
    u}_{s-1}^{n+1}$ provides an approximation of $\mathbf{u}^{n+1}$ of
  order $(\Delta t)^{s-1}$, the overall accuracy order inherent in the
  right-hand side of equation~\eqref{eq:MomentumFactored} is $(\Delta
  t)^{s+1}$, as needed---in view of the $(\Delta t)^2$ prefactor that
  occurs in the expression that contains $\mathbf{\widetilde
    u}_{s-1}^{n+1}$. While the approximation $\mathbf{\widetilde
    u}_{s}^{n+1}$ could have been used while preserving the accuracy
  order, we have found that use of the lower order extrapolation
  $\mathbf{\widetilde u}_{s-1}^{n+1}$ is necessary to ensure
  stability.
\end{remark}

Equation~\eqref{eq:MomentumFactored} can be expressed in the split form
\begin{subequations} \label{eq:DGForm1}
\begin{equation}
\begin{split}
\bigg[ I + b\Delta t \bigg( \widetilde u_s^{n+1} \partial_x 
+ \mu \left( \begin{smallmatrix}
4/3 & 0 \\ 0 & 1
\end{smallmatrix} \right)  \partial_x^2  \bigg) \bigg]  \mathbf u^* 
= & \sum_{j=0}^{s-1} a_j \mathbf u^{n-j}
+ b\Delta t \mu \left( \begin{smallmatrix}
0 & 1/3 \\ 1/3 & 0
\end{smallmatrix} \right) \partial_x\partial_y \mathbf{\widetilde u}_s^{n+1}  \\
 & - b\Delta t \bigg( \widetilde v_s^{n+1}  \partial_y 
+ \mu \left( \begin{smallmatrix}
1 & 0 \\ 0 & 4/3
\end{smallmatrix} \right) \partial_y^2  \bigg) \mathbf{\widetilde u}_{s-1}^{n+1} \label{eq:DGForm1-A}\\
\end{split}
\end{equation}
\begin{equation}
\begin{split}
\bigg[ I + b\Delta t \bigg( \widetilde v_s^{n+1}  \partial_y 
+ \mu \left( \begin{smallmatrix}
1 & 0 \\ 0 & 4/3
\end{smallmatrix} \right) \partial_y^2  \bigg) \bigg] \mathbf u^{n+1} 
= & \sum_{j=0}^{s-1} a_j \mathbf u^{n-j}
+ b\Delta t \mu \left( \begin{smallmatrix}
0 & 1/3 \\ 1/3 & 0
\end{smallmatrix} \right) \partial_x\partial_y \mathbf{\widetilde u}_s^{n+1}  \\
 & - b\Delta t \bigg( \widetilde u_s^{n+1}  \partial_x 
+ \mu \left( \begin{smallmatrix}
4/3 & 0 \\ 0 & 1
\end{smallmatrix} \right)  \partial_x^2  \bigg)
 \mathbf u^* \label{eq:DGForm1-B}
\end{split}
\end{equation}
\end{subequations}
that could be used to evolve the solution from time $t^n$ to time
$t^{n+1}$. We note that these split equations can also be expressed in
the form
\begin{subequations} \label{eq:DGForm2}
\begin{align}
\bigg[ I + b\Delta t \bigg( \widetilde u_s^{n+1} \partial_x 
+ \mu \left( \begin{smallmatrix}
4/3 & 0 \\ 0 & 1
\end{smallmatrix} \right)  \partial_x^2  \bigg) \bigg]  \mathbf u^* 
= & \sum_{j=0}^{s-1} a_j \mathbf u^{n-j} 
+ b\Delta t \mu \left( \begin{smallmatrix}
0 & 1/3 \\ 1/3 & 0
\end{smallmatrix} \right) \partial_x\partial_y \mathbf{\widetilde u}_s^{n+1} \nonumber \\
 & - b\Delta t \bigg( \widetilde v_s^{n+1}  \partial_y 
+ \mu \left( \begin{smallmatrix}
1 & 0 \\ 0 & 4/3
\end{smallmatrix} \right) \partial_y^2  \bigg) \mathbf{\widetilde u}_{s-1}^{n+1}\label{eq:DGForm2a}  \\
\bigg[ I + b\Delta t \bigg( \widetilde v_s^{n+1}  \partial_y 
+ \mu \left( \begin{smallmatrix}
1 & 0 \\ 0 & 4/3
\end{smallmatrix} \right) \partial_y^2  \bigg) \bigg]
\mathbf u^{n+1} 
= & \mathbf u^* + b\Delta t \bigg( \widetilde v_s^{n+1}  \partial_y 
+ \mu \left( \begin{smallmatrix}
1 & 0 \\ 0 & 4/3
\end{smallmatrix} \right) \partial_y^2  \bigg) \mathbf{\widetilde u}_{s-1}^{n+1},
\label{eq:DGForm2b}
\end{align}
\end{subequations}
which is equivalent to~\eqref{eq:DGForm1}---as it can be checked by
subtracting equation~\eqref{eq:DGForm1-A}
from~\eqref{eq:DGForm1-B}. The splitting~\eqref{eq:DGForm2} does not
contain the term involving a differential operator applied to
$\mathbf{u}^* $ on the right-hand side of~\eqref{eq:DGForm1-B} and it
contains, instead, two instances of a term involving a differential
operator applied to $\mathbf{\widetilde u}_{s-1}^{n+1}$. This term
needs to be computed only once for each full time-step and
therefore~\eqref{eq:DGForm2b} leads to a somewhat less expensive
algorithm than~\eqref{eq:DGForm1}.

\section{Unconditional and quasi-unconditional
  stability\label{sec:StabilityDiscussion}}

This section reviews relevant ideas concerning stability in ODE and
PDE theory, and it introduces the new notion of quasi-unconditional
stability.

A variety of stability concepts have been considered in the
literature. Some authors
(e.g.~\cite{gustafsson_time-dependent_2013,leveque_finite_2007})
define (conditional) stability as a requirement of uniform boundedness
of the solution operators arising from spatial and temporal
discretization of the PDE for small enough mesh sizes $h$ and with
$\Delta t\sim h^p$, where $p$ is the order of the underlying spatial
differential operator.  In this paper we will instead define stability
as the boundedness of solutions in terms of their initial data, see
e.g.~\cite{strikwerda_finite_2004}: using a norm $|\cdot|$ which
quantifies the size of the solution at some fixed point $t$ in time,
we say that a scheme is \emph{stable} within some region $\Lambda
\subset \{(h,\Delta t) : h>0,\ \Delta t>0\}$ of
discretization parameter-space if and only if for any final time
$T\in\mathbb{R}$ and all $(h,\Delta t)\in\Lambda$ the estimate
\begin{equation} \label{eq:Stable}
  |Q^n| \leq C_T \sum_{j=0}^{J} |Q^j| 
\end{equation}
holds for all non-negative integers $n$ for which $n\Delta t\leq T$,
where $J$ is a method-dependent non-negative integer, and where $C_T$
is a constant which depends only on $T$.  A method is
\emph{unconditionally stable} for a given PDE problem if $\Lambda$
contains all pairs of positive parameters $(h,\Delta t)$.

The \emph{region of absolute stability} $R$ of an ODE scheme, in turn,
is the set of complex numbers $z = \lambda\,\Delta t$ for which the
numerical solution of the ODE $y'(t) = \lambda\,y(t)$ is stable for
the time-step $\Delta t$. A numerical method which is stable for all
$\Delta t > 0$ and for all $\lambda$ with negative real part is said
to be \emph{A-stable}. In fact, the first- and second-order BDF ODE
solvers are A-stable, and thus may lead to unconditionally stable
methods for certain types of linear PDEs. As is well known, however,
implicit linear multi-step methods of order greater than two, and, in
particular, the BDF schemes of order $s\geq 3$, are not A-stable
(Dahlquist's second barrier~\cite{dahlquist_special_1963}).
Nevertheless, we will see that PDE solvers based on such higher-order
BDF methods may enjoy the property of \emph{quasi-unconditional
  stability}---a concept that we define in what follows.

\begin{definition}\label{def:QuasiUnconditional}
  Let $\Omega_h$ be a family of spatial discretizations of a domain
  $\Omega$ controlled by a mesh-size parameter $h$ and let $\Delta t$
  be a temporal step size. A numerical method for the solution of the
  PDE $Q_t = \mathcal{P}\,Q$ in $\Omega$ is said to be
  \textbf{\emph{quasi-unconditionally stable}} if there exist positive
  constants $M_h$ and $M_t$ such that the method is stable for all $h
  < M_h$ and all $\Delta t < M_t$.
\end{definition}

Clearly, quasi-unconditional stability implies that for small enough
$\Delta t$, the method is stable for arbitrarily fine spatial
discretizations. Note that  stability may still take
place outside of the  quasi-unconditional
stability rectangle $(0,M_h)\times(0,M_t)$ provided additional stability constraints are satisfied.  For
example, Figure~\ref{fig:QuasiUnconStab} presents a schematic of the
stability region for a notional method that enjoys quasi-unconditional
stability in the parameter space $(h,\Delta t)$ as well as conditional
CFL-like stability outside the quasi-unconditional stability
rectangle. In practice we have encountered quasi-unconditionally
stable methods whose stability outside the window
$(0,M_h)\times(0,M_t)$ is delimited by an approximately straight curve
similar to that displayed in Figure~\ref{fig:QuasiUnconStab}.

\begin{figure}[!htb]
	\centering
	\includegraphics[width=0.5\textwidth]{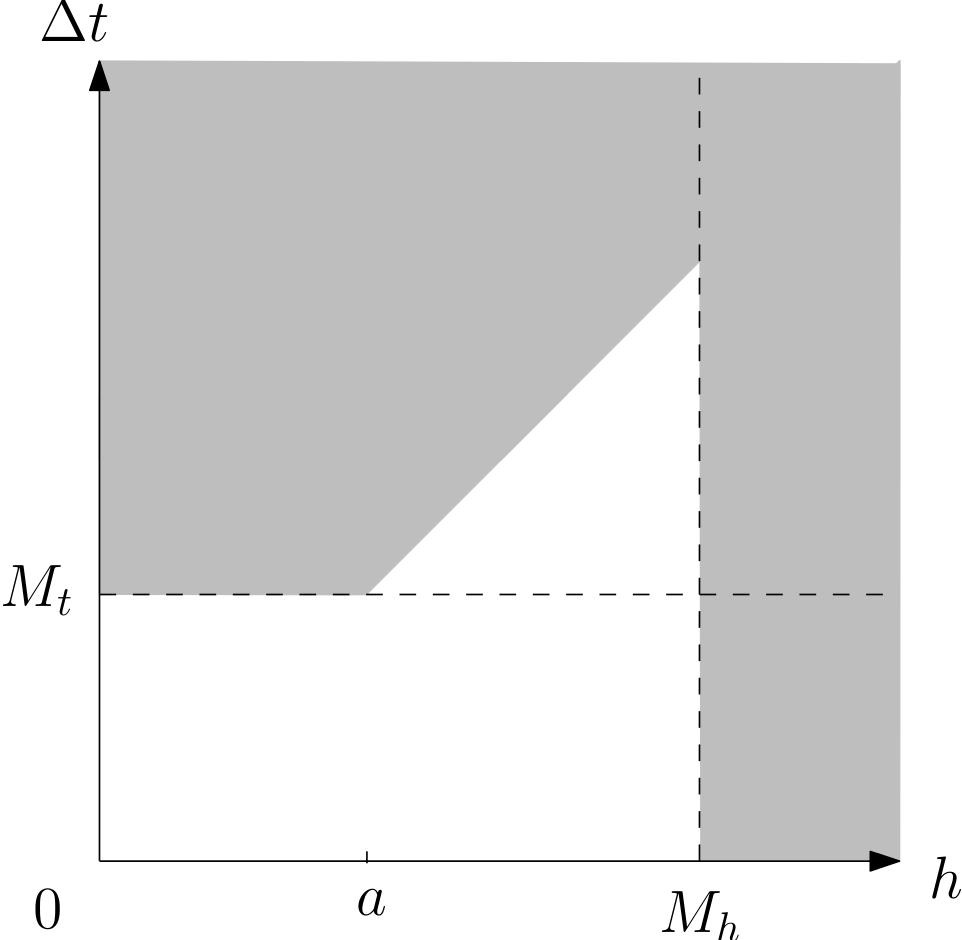}
        \caption{Stability region of a notional quasi-unconditionally
          stable method. The white region is the set of $(h,\Delta t)$
          pairs for which the method is stable.  Notice that outside
          of the quasi-unconditional stability rectangle
          $(0,M_h)\times(0,M_t)$, the present hypothetical method is
          stable for time steps satisfying the conditions $\Delta t <
          \min\{h,M_h\}$; other types of CFL-like conditions do, of
          course, occur commonly in practice. Thus, in particular,
          quasi-unconditional stability does not exclude the
          possibility of stability outside the rectangle
          $(0,M_h)\times(0,M_t)$.}
	\label{fig:QuasiUnconStab}
\end{figure}

In lieu of a full stability analysis for the main problem under
consideration (the fully non-linear compressible Navier-Stokes
equations, for which stability analyses are not available for any of
the various extant algorithms), in support of the stability behavior
observed in our numerical experiments we present rigorous stability
results for simpler related problems. In particular,
Section~\ref{sec:BDF2Stability} establishes the unconditional
stability of the Fourier-based BDF2-ADI scheme for linear constant
coefficient hyperbolic and parabolic equations in two spatial
dimensions. Section~\ref{sec:HighOrderBDFStability}, in turn, shows
that quasi-unconditional stability takes place for Fourier-spectral
BDF methods of order $s$ ($2 \leq s \leq 6$, without ADI) for the
advection-diffusion equation in one- and two-dimensional space, and
Section~\ref{sec:BDFLinearNS} presents numerical tests that
demonstrate quasi-unconditional stability for the full compressible
Navier-Stokes equations.

\begin{remark} \label{rem:vonNeumann} In general, the stability of a
  PDE solver can be ensured provided relevant discrete operators are
  power bounded~\cite{trefethen_pseudospectra_1997}. The von Neumann
  criterion provides a necessary but not sufficient condition for the
  power-boundedness of solution operators. If the discrete operators
  are non-normal, then stability analysis requires application of the
  Kreiss matrix theorem~\cite[p. 177]{gustafsson_time-dependent_2013}.
  In particular, it is
  known~\cite{trefethen_pseudospectra_1997,trefethen_spectra_2005}
  that certain discretizations and numerical boundary conditions can
  give rise to non-normal families of solution operators that are not
  power-bounded (and unstable) even though the underlying problem is
  linear with constant coefficients and all eigenvalues are inside the
  unit disk. (An operator $P$ with adjoint $P^*$ is normal if $P^*P =
  PP^*$.) In our Fourier-spectral context, however, all operators are
  normal (which follows from the fact that the first derivative
  operators are skew-Hermitian, the second derivative operators are
  Hermitian, and all derivative operators commute) and consequently
  the von Neumann criterion is both necessary and sufficient
  (cf.~\cite[p.  189]{hesthaven_spectral_2007}).
\end{remark}

\section{Unconditional stability for BDF2-ADI: periodic linear case
   \label{sec:BDF2Stability}}

 In this section the energy method is used to establish the
 unconditional stability of the BDF2-ADI scheme for the constant
 coefficient hyperbolic and parabolic equations with periodic boundary
 conditions under a Fourier collocation spatial approximation
 (Sections~\ref{sec:Hyperbolic} and~\ref{sec:Parabolic}). In the case
 of the parabolic equation, further, an essentially identical argument
 is used in Section~\ref{sec:LegendreParabolic} to establish the
 corresponding unconditional stability result for the Legendre
 polynomial spectral collocation method with (non-periodic)
 homogeneous boundary conditions. Unfortunately, as discussed in that
 section, such a direct extension to the non-periodic case has not
 been obtained for the hyperbolic equation.

\subsection{Preliminary definitions\label{sec:PrelimDefs}} 

We consider the domain
\begin{equation}\label{eq:PeriodicDomain} 
  \Omega = [0,2\pi) \times [0,2\pi)
\end{equation}
which we discretize on the basis of an odd number $N+1$ of
discretization points ($N$ even, for definiteness) in both the $x$ and
$y$ directions ($x_{j} = 2\pi j/(N+1)$ and $y_{k} = 2\pi k/(N+1)$,
$0\leq j,k \leq N$), which gives rise to the grid
\begin{equation}\label{eq:Grid}
  \{(x_j,y_k):0\leq j,k \leq N\}.
\end{equation}
(The restriction to even values of $N$, which is introduced for
notational simplicity, allows us to avoid changes in the form of the
summation limits in the Fourier
series~\eqref{eq:FourierInterpolant}. Similarly, our use of equal
numbers of points in the $x$ and $y$ directions simplifies the
presentation somewhat. But, clearly, extensions of our constructions
that allow for odd values of $N$ as well as unequal numbers of points
in the $x$ and $y$ direction are straightforward.)

For (complex valued) grid functions
\begin{equation} \label{eq:GridFunctions}
f = \{f_{jk}\} \text{ and } g = \{g_{jk}\},\quad 0\leq j,k \leq N
\end{equation}
we define the discrete inner product and norm
\begin{equation} \label{eq:InnerProduct}
(f,g) = \frac{1}{(N+1)^2}\sum\limits_{j,k} f_{jk} \bar g_{jk} 
\end{equation}
$$ |f| = \sqrt{(f,f)}. $$
Each grid function $f$ as in~\eqref{eq:GridFunctions} is can be
associated to a trigonometric interpolant $f_N(x,y)$ ($f_N(x_j,y_k) =
f_{jk}$) which is given by
\begin{equation} \label{eq:FourierInterpolant}
  f_N(x,y) = \sum_{|j|,|k| \leq \frac{N}{2}} \widehat{f}_{jk} e^{i(jx+ky)}
\end{equation}
where 
$$ \widehat{f}_{jk} = \frac{1}{(N+1)^2} \sum\limits_{j,k} f_{jk} e^{-i(jx_{j}+ky_{k})}. $$
Note that the inner product~\eqref{eq:InnerProduct} coincides with the
trapezoidal quadrature rule applied to the grid functions $f$ and $g$
over the underlying domain $[0,2\pi)\times[0,2\pi)$. Since the
trapezoidal rule is exact for all truncated
Fourier series containing exponentials of the form $e^{-i(jx+ky)}$
with $-N\leq j,k\leq N$, it follows that the discrete inner
product~\eqref{eq:InnerProduct} equals the integral inner product of
the corresponding trigonometric interpolants---i.e.
\begin{equation} \label{eq:InnerProductToIntegral} 
  (f,g) = \frac{1}{(2\pi)^2} \int_0^{2\pi} \int_0^{2\pi} f_N(x,y) \bar g_N(x,y) \,dx\,dy.
\end{equation}

In order to discretize solutions of PDEs we utilize {\em time
  sequences of grid functions} $u = \{u^n: n \geq 0\}$, where, for
each $n$, $u^n = \{u^n_{jk}\}$ is a grid function such as those
displayed in equation~\eqref{eq:GridFunctions}. For such time series
the scalar product~\eqref{eq:InnerProduct} at fixed $n$ can be used to
produce a time series of scalar products: the inner product of two
time series of grid functions $u=\{u^n: n\geq 0 \}$ and $v=\{v^n:
n\geq 0 \}$ is thus a time series of complex numbers:
$$
  (u,v) = \{ (u^n,v^n): n\geq 0 \}.
$$

\subsection{Discrete spatial and temporal operators\label{sec:DiscreteOperators}} 

In order to discretize PDEs we use discrete spatial and temporal
differentiation operators that act on grid functions and time-series,
respectively.

We consider spatial differentiation first: the Fourier $x$-derivative
operator $\delta_x$ applied to a grid function $f$, for example, is
defined as the grid function $\delta_x f$ whose $jk$ value equals the
value of the derivative of the interpolant $f_N$ at the point
$(x_{j},y_k)$:
\begin{equation} \label{eq:DerivativeExact} 
  (\delta_x f)_{jk} = \frac{\partial}{\partial x} f_N(x_{j},y_{k}). 
\end{equation}
The operators $\delta_{xx}$, $\delta_y$, $\delta_{yy}$, $\delta_{xy}=
\delta_{x}\delta_{y} = \delta_{y}\delta_{x}$ etc. are defined
similarly.

Using the exactness relation~\eqref{eq:InnerProductToIntegral} and integration
by parts together with the periodicity of the domain, it follows that the
first derivative operators $\delta_x$ and $\delta_y$ are skew-Hermitian and the
second derivative operators $\delta_{xx}$, $\delta_{yy}$ are Hermitian:
%
\begin{equation} \label{eq:DeltaXAntisym} 
  (\delta_x f,g) = -(f,\delta_x g),\;\;(\delta_y f,g) = -(f,\delta_y g),\;\;(\delta_{xx} f,g) = (f,\delta_{xx} g),\;\;(\delta_{yy} f,g) = (f,\delta_{yy} g).
\end{equation}
Certain temporal differentiation and extrapolation operators we use,
in turn, produce a new time series for a given time series---for both
numerical time series as well as time series of grid functions. These
operators include the regular first and second order finite difference
operators $D$ and $D^2$, the three-point backward difference operator
$\widehat{D}$ that is inherent in the BDF2 algorithm, as well as the
second order accurate extrapolation operator ``$\sim$'':
\begin{alignat}{3}
  (Du)^n &= u^n - u^{n-1} \quad && n \geq 1 \\
  (D^2 u)^n &= (Du)^n - (Du)^{n-1} =  u^n - 2u^{n-1} + u^{n-2} \quad && n \geq 2 \\
  (\widehat{D} u)^n &= \frac{3}{2} u^n - 2 u^{n-1} + \frac{1}{2} u^{n-2}  \quad && n \geq 2 \\
  \widetilde{u}^{n+1} &= 2u^n - u^{n-1}  \quad && n \geq 1 .
\end{alignat}
Note that the members of the time series $\widehat{D} u$ can also be
expressed in the forms
\begin{align}
	(\widehat{D} u)^n &= D \left( u^n + \frac{1}{2}(Du)^n \right) \label{eq:DuHatFormDDu} \\
			    &= \frac{1}{2} ( (Du)^n + (D\widetilde{u})^{n+1} ) \label{eq:DuHatFormTilde} \\
			    &= \frac{3}{2} (Du)^n - \frac{1}{2} (Du)^{n-1}.
\end{align}

In what follows we make frequent use of the finite difference product
rule for two time series $u$ and $v$:
\begin{equation} \label{eq:FDProductRule}
  u\,Dv = D(uv) - v\,Du + (Du)\,(Dv).
\end{equation}
An immediate consequence of~\eqref{eq:FDProductRule}, which will also
prove useful, concerns the real part of scalar products of the form
$(Du,P\,u)$ where $P$ is an operator which is self-adjoint with
respect to the discrete inner product~\eqref{eq:InnerProduct} and
which commutes with $D$. For such operators we have the identity
\begin{equation}\label{eq:ProductRule}
\Re (Du,P\,u) = \frac{1}{2} D(u,P\,u) + \frac{1}{2} (Du,P\,Du)
\end{equation}
which follows easily from the relations
\begin{align}
	(Du,P\,u) &= D(u,P\,u) - (u,DP\,u) + (Du,DP\,u) \nonumber \\
	          &= D(u,P\,u) - (P\,u,Du) + (Du,P\,Du) \nonumber \\
	          &= D(u,P\,u) - \overline{(Du,P\,u)} + (Du,P\,Du). \nonumber 
\end{align}
%

\subsection{Periodic BDF2-ADI stability: hyperbolic equation\label{sec:Hyperbolic}}

This section establishes the unconditional stability of the BDF2-ADI
method for the constant-coefficient advection equation
\begin{equation}\label{eq:ConvectionEq} 
U_t + \alpha U_x + \beta U_y = 0
\end{equation}
in the domain~\eqref{eq:PeriodicDomain}, with real constants $\alpha$
and $\beta$, and subject to periodic boundary conditions. The BDF2-ADI
scheme for the advection equation can be obtained easily by adapting
the corresponding form~\eqref{eq:MomentumFactored} of the BDF2-ADI
scheme for the pressure-free momentum equation. Indeed, using the
Fourier collocation approximation described in the previous two
sections, letting $u$ denote the discrete approximation of the
solution $U$, and noting that, in the present context the necessary
extrapolated term $\mathbf{\widetilde u}_{s-1}^{n+1}$ in
equation~\eqref{eq:MomentumFactored} equals $u^n$, the factored form
of our BDF2-ADI algorithm for equation~\eqref{eq:ConvectionEq} is
given by
\begin{equation}\label{eq:BDF2Convection}
  (I+b\Delta t \alpha \delta_x)(I+b\Delta t \beta \delta_y)u^{n+1} = a_0 u^n +a_1 u^{n-1} + \alpha\beta(b\Delta t)^2\delta_x\delta_yu^n.
\end{equation}

Before proceeding to our stability result we derive a more convenient
(equivalent) form for equation~\eqref{eq:BDF2Convection}: using the
numerical values $a_0=4/3$, $a_1=-1/3$, and $b=2/3$ of the BDF2
coefficients (see e.g.~\cite[Ch. 3.12]{lambert_numerical_1991}
or~Part~I), the manipulations
\begin{align*}
  0 &= (I+b\Delta t \alpha \delta_x)(I+b\Delta t \beta \delta_y)u^{n+1} - a_0 u^n - a_1 u^{n-1} - \alpha\beta(b\Delta t)^2\delta_x\delta_yu^n \\
    &= u^{n+1} - a_0 u^n - a_1 u^{n-1} + b\Delta t \alpha \delta_x u^{n+1} + b\Delta t \beta \delta_y u^{n+1} + \alpha\beta(b\Delta t)^2\delta_x\delta_y(u^{n+1} - u^n) \\
    &= \frac{1}{b}(u^{n+1} - a_0 u^n - a_1 u^{n-1}) + \Delta t \alpha \delta_x u^{n+1} + \Delta t \beta \delta_y u^{n+1} + b\alpha\beta(\Delta t)^2\delta_x\delta_y(u^{n+1} - u^n)
\end{align*}
reduce equation~\eqref{eq:BDF2Convection} to the form
\begin{equation} \label{eq:DiscreteHyperbolicAB}
	\widehat{D} u + A\, u + B\, u + bAB\, Du = 0,
\end{equation}
where $b= 2/3$, $A=\alpha\Delta t \delta_x $ and $B=\beta\Delta t
\delta_y$.

We are now ready to establish an energy stability estimate for the
BDF2-ADI equation~\eqref{eq:BDF2Convection}.
\begin{theorem} \label{thm:HyperbolicStability} 
The solution $u$ of~\eqref{eq:BDF2Convection} with initial conditions $u^0$ and
$u^1$ satisfies
$$ |u^n|^2 + |\widetilde{u}^{n+1}|^2 +  \frac{2}{3} \left( |A u^n|^2 + |B u^n|^2 +\sum_{m=2}^n |(D^2 u)^m|^2 \right) \leq M $$
for all $n \geq 2$, where
$$ M = |u^1|^2 + |\widetilde{u}^{2}|^2 + \frac 23 ( |A u^1|^2 + |B u^1|^2 ). $$
In particular, the scheme is unconditionally stable in the sense that
a bound of the form given in equation~\eqref{eq:Stable} is satisfied.
\end{theorem}

\begin{proof}
Taking the inner product of equation~\eqref{eq:DiscreteHyperbolicAB}
with $u$ we obtain
\begin{alignat}{7} \label{eq:uProductHyperbolic}
  0 &= (u,\widehat Du) && + (u,A\,u) && + (u,B\,u) && + b(u,AB\,Du) \\
  &= (\mathit{I}) && + (\mathit{II}) && + (\mathit{III}) && + (\mathit{IV}), \nonumber
\end{alignat}
where $(I) =(u,\widehat Du)$, $(\mathit{II})=(u,A\,u)$, etc. Our goal is
to express the real part of the right-hand side
in~\eqref{eq:uProductHyperbolic} as a sum of non-negative terms and
telescoping terms of the form $Df$ for some non-negative numerical time
series $f$. To that end, we consider the terms $(I)$ through
$(\mathit{IV})$ in turn.

\paragraph{$(\mathit{I})$:} Using the
expression~\eqref{eq:DuHatFormTilde} for $\widehat Du$ we obtain
\begin{equation} \label{eq:uDHatU}
  (I) = \frac{1}{2} (u,Du) + \frac{1}{2} (u,D\widetilde w),
\end{equation}
where $\widetilde w$ denotes the time series obtained by shifting $\widetilde u$
forwards by one time step:
\begin{equation} \label{eq:WDefinition}
  \widetilde w = \{ \widetilde w^n = \widetilde u^{n+1} : n\geq 1 \}.
\end{equation}
To re-express~\eqref{eq:uDHatU} we first note that for any two
grid functions $a$ and $b$ we have the relation
\begin{align*}
           & |a-b|^2 = |a|^2 + |b|^2 - 2\Re (a,b) \\
  \implies & \Re (a,b) = \frac{1}{2}(|a|^2 + |b|^2 - |a-b|^2).
\end{align*}
Therefore, for any time series $g$ we have
\begin{align}
  \Re (u,Dg)^n &= \Re (u^n,g^n) - \Re (u^n,g^{n-1}) \nonumber \\
         &= \frac{1}{2} ( |u^n|^2 + |g^n|^2 - |u^n-g^n|^2 ) - \frac{1}{2} ( |u^n|^2 + |g^{n-1}|^2 - |u^n-g^{n-1}|^2 ) \nonumber \\
         &= \frac{1}{2}( D|g^n|^2 - |u^n-g^n|^2 + |u^n-g^{n-1}|^2 ). \label{eq:uDg}
\end{align}
Letting $g=u$ and $g=\widetilde w$ in~\eqref{eq:uDg} we obtain
\begin{equation} \label{eq:uDHatUPart1}
  \Re (u,Du) = \frac{1}{2}( D|u|^2 + |Du|^2 ) \\
\end{equation}
and
\begin{equation} \label{eq:uDHatUPart2}
  \Re(u,D\widetilde w) = \frac{1}{2}( D|\widetilde w|^2 - |Du|^2 + |D^2u|^2 ).
\end{equation}
Replacing~\eqref{eq:uDHatUPart1} and~\eqref{eq:uDHatUPart2}
in~\eqref{eq:uDHatU} we obtain
\begin{equation} \label{eq:uDHatUBound}
  \Re (\mathit{I}) = \frac{1}{4}D(|u|^2+|\widetilde w|^2) + \frac{1}{4}|D^2u|^2.
\end{equation}
Notice that this equation expresses $\Re (\mathit{I})$ as the sum of a
telescoping term and a positive term, as desired.

\paragraph{$(\mathit{II})$ and $(\mathit{III})$:}

The operator $A$ is clearly skew-Hermitian since $\delta_x$ is. Therefore
\begin{align}
  & (\mathit{II}) = (u,Au) = -(Au,u) = -\overline{(u,Au)} \nonumber \\
  \implies & \Re(\mathit{II}) = 0. \label{eq:uAuZero}
\end{align}
The relation
\begin{equation} \label{eq:uBuZero}
  \Re (\mathit{III}) = \Re (u,Bu) = 0
\end{equation}
follows similarly, of course.

\paragraph{$(\mathit{IV})$:}
Lemma~\ref{lem:uABDuBound} below tells us that
\begin{equation} \label{eq:uABDuBoundThm}
   \Re (u,AB\,Du) \geq \frac{1}{4} D \left( |Au|^2 + |Bu|^2 \right) - \frac{1}{8} |D^2u|^2.
\end{equation}
Substituting~\eqref{eq:uDHatUBound}, \eqref{eq:uAuZero},
\eqref{eq:uBuZero}, and~\eqref{eq:uABDuBoundThm} into
equation~\eqref{eq:uProductHyperbolic} (recalling $b=2/3$) and taking
the real part we obtain
\begin{align}
	0 \geq \frac{1}{4}D \left( |u|^2+|\widetilde w|^2 \right) + \frac 16 \left( |Au|^2 + |Bu|^2 + |D^2u|^2 \right),
\end{align}
which is the sum of a telescoping term and a non-negative
term. Multiplying by the number four and summing the elements of the
above numerical time series from $m=2$ to $n$ completes the proof of
the theorem.
\end{proof}

The following lemma concerns the bound~\eqref{eq:uABDuBoundThm} used
in Theorem~\ref{thm:HyperbolicStability}.
\begin{lemma} \label{lem:uABDuBound}
  Any solution of equation~\eqref{eq:DiscreteHyperbolicAB}
  satisfies~\eqref{eq:uABDuBoundThm}.
\end{lemma}

\begin{proof}
  Taking the inner product of~\eqref{eq:DiscreteHyperbolicAB} with
  $A\, Du$ (using the form~\eqref{eq:DuHatFormDDu} of $\widehat D u$) we obtain
\begin{equation} \label{eq:ADuProduct}
	0 = (Du,A\,Du) + \frac{1}{2}(D^2u,A\,Du) + (A\,u,A\,Du) + (B\,u,A\,Du) + b(AB\,Du,A\,Du).
\end{equation}
Since $A$ and $B$ commute and since $B$ is skew-Hermitian
(equation~\eqref{eq:DeltaXAntisym}) we have
\begin{align*}
  (B\,u,A\,Du) &= -(u,AB\,Du) 
\end{align*}
for the next-to-last term in~\eqref{eq:ADuProduct}. Therefore,
equation~\eqref{eq:ADuProduct} can be re-expressed in the form
\begin{alignat}{9} 
	(u,AB\,Du) &= (Du,A\,Du) && + \frac{1}{2}(D^2u,A\,Du) && + (A\,u,A\,Du)   && + b(AB\,Du,A\,Du) \label{eq:ADuProduct2} \\
	           &= (\mathit{I})        && + (\mathit{II})           && + (\mathit{III}) && + (\mathit{IV}). \nonumber
\end{alignat}
We consider each term in~\eqref{eq:ADuProduct2} in turn.

\paragraph{$(\mathit{I})$:}

Since $A$ is skew-Hermitian it follows that the real part of this term
vanishes:
\begin{align}
  & (\mathit{I}) = (Du,A\,Du) = -(A\,Du,Du) = -\overline{(Du,A\,Du)}= -\overline{(\mathit{I})}\nonumber \\
  \implies & \Re(\mathit{I}) = 0. \label{eq:uABDuPart1}
\end{align}

\paragraph{$(\mathit{II})$:}

Using Young's inequality 
\begin{equation}\label{eq:YoungIneq}
  ab \leq \frac r2 a^2 + \frac{1}{2r}b^2 
\end{equation} 
(which, as is easily checked, is valid for all real numbers $a$ and
$b$ and for all $r>0$) together with the Cauchy-Schwarz inequality we
obtain
\begin{align}
	\Re (\mathit{II}) &= \frac{1}{2}\Re (D^2u,A\,Du) \nonumber \\
	     &\geq -\frac{1}{2}|(D^2u,A\,Du)| \nonumber \\
	     &\geq -\frac{1}{2}|D^2u|\,|A\,Du| \nonumber \\
	     &\geq -\frac{1}{2}(\frac{1}{4}|D^2u|^2  + |A\,Du|^2 ) \nonumber \\
	     &= - \frac{1}{8}|D^2u|^2 - \frac{1}{2}|A\,Du|^2. \label{eq:uABDuPart2} 
\end{align}

\paragraph{$(\mathit{III})$:}

By the finite-difference product rule~\eqref{eq:FDProductRule} we obtain
\begin{alignat}{3}
	& \mathit{(III)} &&= (Au,D(Au)) \nonumber \\
	         &       &&= D(Au,Au) - (DAu,Au) + (DAu,DAu) \nonumber \\
	         &       &&= D|Au|^2 - \overline{(III)} + |A\,Du|^2 \nonumber \\
	\implies & \Re (III) &&= \frac{1}{2} D|Au|^2 + \frac{1}{2} |A\,Du|^2. \label{eq:uABDuPart3} 
\end{alignat}

\paragraph{$(\mathit{IV})$:}

Again using the fact that $B$ is skew-Hermitian and commutes with $A$ it follows that
\begin{align}
  & (\mathit{IV}) = b(BA\,Du,A\,Du) = -b(A\,Du,BA\,Du) = -\overline{\mathit{(IV)}} \nonumber \\
  \implies & \Re(\mathit{IV}) = 0. \label{eq:uABDuPart4}
\end{align}
Combining the real parts of
equations~\eqref{eq:ADuProduct2},~\eqref{eq:uABDuPart1},
\eqref{eq:uABDuPart2}, \eqref{eq:uABDuPart3} and~\eqref{eq:uABDuPart4} we obtain
\begin{equation} \label{eq:uABDuBoundA}
   \Re (u,AB\,Du) \geq \frac{1}{2} D|Au|^2 - \frac{1}{8} |D^2u|^2.
\end{equation}
An analogous result can be obtained by taking the inner product of
equation~\eqref{eq:DiscreteHyperbolicAB} with $B\,Du$ instead of $A\,Du$ and
following the same steps used to arrive at equation~\eqref{eq:uABDuBoundA}. The
result is
\begin{equation} \label{eq:uABDuBoundB}
   \Re (u,AB\,Du) \geq \frac{1}{2} D|Bu|^2 - \frac{1}{8} |D^2u|^2.
\end{equation}
The lemma now follows by averaging equations~\eqref{eq:uABDuBoundA}
and~\eqref{eq:uABDuBoundB}.
\end{proof}

\subsection{Fourier-based BDF2-ADI stability: parabolic equation\label{sec:Parabolic}}

This section establishes the unconditional stability of the BDF2-ADI
method for the constant-coefficient parabolic equation
\begin{equation}\label{eq:ParabolicEq}
U_t = \alpha\, U_{xx} + \beta\, U_{yy} + \gamma\, U_{xy}.
\end{equation}
Notice the inclusion of the mixed derivative term, which is treated explicitly using temporal extrapolation in the BDF-ADI algorithm. Theorem~\ref{thm:ParabolicStability} in this section proves, in particular, that extrapolation of the mixed derivative does not compromise the unconditional stability of the method.

The parabolicity conditions  $\alpha > 0$, $\beta > 0$ and
\begin{equation} \label{eq:ParabolicCondition}
	\gamma^2 \leq 4 \alpha \beta,
\end{equation}
which are assumed throughout this section, ensure that
\begin{equation}\label{eq:IntegralParabolicRelation}
  \int_0^{2\pi} \int_0^{2\pi} f \left( \alpha\, f_{xx} + \beta\, f_{yy} + \gamma\, f_{xy} \right) \,dx\,dy \leq 0
\end{equation}
%
for any twice continuously differentiable bi-periodic function $f$
defined in the domain~\eqref{eq:PeriodicDomain}---as can be established
easily by integration by parts and completion of the square in the sum
$\alpha (f_x)^2 + \gamma f_xf_y$ together with some simple
manipulations.  In preparation for the stability proof that is put
forth below in this section, in what follows we present a few
preliminaries concerning the BDF2-ADI algorithm for
equation~\eqref{eq:ParabolicEq}.

We first note that a calculation similar to that leading to
equation~\eqref{eq:DiscreteHyperbolicAB} shows that the Fourier-based
BDF2-ADI scheme for~\eqref{eq:ParabolicEq} can be expressed in the form
\begin{equation} \label{eq:DiscreteParabolic}
\widehat{D} u - \Delta t ( \alpha\, \delta_{xx} + \beta\, \delta_{yy} u + \gamma\, \delta_x \delta_y ) u + \Delta t \, \gamma\, \delta_x \delta_y \, D^2u + b(\Delta t)^2 \alpha\beta\, \delta_{xx} \delta_{yy} \, Du = 0.
\end{equation}
Letting
\begin{align*}
&A = -\Delta t\,\alpha\,\delta_{xx}, \\
&B = -\Delta t\,\beta\,\delta_{yy}, \\
&F = -\Delta t\,\gamma\,\delta_x\delta_y, \\
&L = A+B+F,
\end{align*}
equation~\eqref{eq:DiscreteParabolic} becomes
\begin{equation} \label{eq:DiscreteParabolic2}
\widehat{D} u + Lu - F\,D^2u + bAB\,Du = 0.
\end{equation}
Note that the operators $A$ and $B$ above do not coincide with the
corresponding $A$ and $B$ operators in Section~\ref{sec:Hyperbolic}.

In view of the exactness relation~\eqref{eq:InnerProductToIntegral}
together with the Fourier differentiation operators
(cf.~\eqref{eq:DerivativeExact}), it follows that
$A$, $B$, $AB$ and $L$ are positive semidefinite operators. Indeed, in
view of equation~\eqref{eq:IntegralParabolicRelation}, for example, we have
\begin{align} \label{eq:PosDef}
	(u,L\,u) &= -\frac{\Delta t}{(2\pi)^2} \int_0^{2\pi} \int_0^{2\pi} u_N \, \left( \alpha (u_N)_{xx} + \beta (u_N)_{yy} + \gamma (u_N)_{xy} \right) \,dx\,dy \nonumber \\
        &\geq 0;
\end{align}
similar relations for $A$, $B$ and $AB$ follow directly by integration
by parts.

Finally we present yet another consequence of the parabolicity
condition~\eqref{eq:ParabolicCondition} which will prove useful: for
any grid function $g$ we have
\begin{equation}
  |Fg|^2 = \gamma^2\,(\Delta t)^2 (\delta_x\delta_y g,\delta_x\delta_y g) 
  \leq 4 \alpha\beta(\Delta t)^2 (g, \delta_x^2 \delta_y^2 g)= 4 (g,ABg). 
\end{equation}
Thus, defining the seminorm
\begin{equation}\label{eq:Seminorm}
  |u|_P = \sqrt{(u,Pu)}
\end{equation}
%
for a given positive semidefinite operator $P$ and using $P = AB$
we obtain
\begin{equation}
  |Fg|^2 \leq 4 |g|^2_{AB}. \label{eq:FtoAB}
\end{equation}

The following theorem can now be established.
\begin{theorem} \label{thm:ParabolicStability} 
The solution $u$ of the
Fourier-based BDF2-ADI scheme ~\eqref{eq:DiscreteParabolic} for
equation~\eqref{eq:ParabolicEq} with initial conditions $u^0$, $u^1$
satisfies
$$ \frac{1}{4}|u^n|^2 + \frac{1}{4}|\widetilde{u}^{n+1}|^2 + \frac{1}{3} |(Du)^n|_{AB}^2 + \frac{1}{4}\sum_{m=1}^n |D^2u|^2 + \sum_{m=1}^n |u^n|_L^2 \leq M $$
for $n \geq 2$, where
\begin{align*}
	M =& \frac{1}{4}|u^1|^2 + \frac{1}{4}|\widetilde{u}^{2}|^2 + \frac{1}{3} |u^1|_{AB}^2 + 3 |u^1|_L - \Re(u^1,F\,(Du)^1)  \\
     & {} + 3|(Du)^1|^2 + \frac 32 \left( |(Du)^1|_A^2 + |(Du)^1|_B^2 \right) + \frac{1}{3} |(Du)^1|_{AB}^2.
\end{align*}
In particular, the scheme is unconditionally stable in the sense that
a bound of the form given in equation~\eqref{eq:Stable} is satisfied.
\end{theorem}

\begin{proof}
Taking the inner product of~\eqref{eq:DiscreteParabolic2} with $u$ we obtain 
\begin{alignat}{7} \label{eq:UParabolicProduct}
	0 &= (u,\widehat{D}u) && + (u,Lu) && - (u,F\,D^2u) && + b(u,AB\,Du) \\
	  &= (I)          && + (\mathit{II})   && + (\mathit{III})       && + (\mathit{IV}), \nonumber
\end{alignat}
where $(I) = (u,\widehat Du)$, $(\mathit{II}) = (u,Lu)$, etc. As in
Theorem~\ref{thm:HyperbolicStability}, we re-express the above equation using
telescoping and non-negative terms to obtain the desired energy bound.

The term $(I)$ already occurs in the proof of
Theorem~\ref{thm:HyperbolicStability}; there we obtained the relation
\begin{equation} \label{eq:uDHatUBound2}
  \Re (\mathit{I}) = \frac{1}{4}D(|u|^2+|\widetilde w|^2) + \frac{1}{4}|D^2u|^2,
\end{equation}
where $\widetilde w$ is defined in~\eqref{eq:WDefinition}. The term
$(\mathit{II}) = |u|_L^2$, in turn, is non-negative (see
equation~\eqref{eq:PosDef}) and thus requires no further
treatment. The remaining two terms are considered in what follows.

\paragraph{$(\mathit{III})$:}

This term presents the most difficulty, since $F$ is not positive
semi-definite. In what follows the term ($\mathit{III}$) is
re-expressed as a a sum of two quantities, the first one of which can
be combined with a corresponding term arising from the quantity
$\mathit{(IV)}$ to produce a telescoping term, and the second of which
will be addressed towards the end of the proof by utilizing
Lemma~\ref{lem:DuParabolicBound} below.

Let $v$ denote the time series obtained by shifting $u$
backwards by one time step:
\begin{equation} \label{eq:VDefinition}
  v = \{ v^n = u^{n-1} : n\geq 1 \};
\end{equation}
clearly we have
\begin{equation}\label{eq:DuvDef} 
  Du = u - v\quad \mbox{and}\quad D^2u = Du - Dv.
\end{equation}
Thus, using the finite difference product
rule~\eqref{eq:FDProductRule} and the second relation in~\eqref{eq:DuvDef}
we obtain
\begin{align*}
	(\mathit{III}) &= -(u,F\,D(Du))=-(u,D\,F(Du)) \\
	      &= -D(u,F\,Du) + (Du,F\,Du) - (Du,F\,D^2u)  \\
	      &= -D(u,F\,Du) + (Du,F\,Dv). 
\end{align*}
Applying the Cauchy-Schwarz inequality and Young's
inequality~\eqref{eq:YoungIneq} with $r=6$ together with~\eqref{eq:FtoAB} we
obtain
\begin{align}
  \Re (\mathit{III}) &\geq -D\Re (u,F\,Du) - |Du|\,|F\,Dv| \nonumber \\
        &\geq -D\Re(u,F\,Du) - 3|Du|^2 - \frac{1}{12} |F\,Dv|^2 \nonumber \\
        &\geq -D\Re(u,F\,Du) - 3|Du|^2 - \frac{1}{3} |Dv|_{AB}^2. \label{eq:UParabolicProductPart3}
\end{align}
The last term in the above inequality will be combined with an
associated expression in $(\mathit{IV})$ below to produce a
telescoping term.

\paragraph{$(\mathit{IV})$:}

Using the finite difference product rule~\eqref{eq:ProductRule}
together with the fact that $AB$ is a Hermitian positive semi-definite
operator we obtain
\begin{align}
  \Re (\mathit{IV}) &= \frac 23 \Re (u,AB\,Du)= \frac 23 \Re (Du,ABu) \nonumber \\
  &= \frac{1}{3} D(u,AB\,u) + \frac{1}{3} (Du,AB\,Du) \nonumber \\
  &= \frac{1}{3} D|u|_{AB}^2 + \frac{1}{3} |Du|_{AB}^2 \label{eq:UParabolicProductPart4}
\end{align}
(see equation~\eqref{eq:Seminorm}).
Substituting~\eqref{eq:uDHatUBound2},
\eqref{eq:UParabolicProductPart3} and
\eqref{eq:UParabolicProductPart4} into
equation~\eqref{eq:UParabolicProduct}, recalling
equation~\eqref{eq:WDefinition} and taking real parts, we obtain
\begin{align} \label{eq:UParabolicProduct2}
  0 \geq& \frac{1}{4}D(|u|^2+|\widetilde w|^2) + \frac{1}{4}|D^2u|^2 + |u|_L^2 - D\Re(u,F\,Du) - 3|Du|^2 \nonumber \\
        & {} + \frac{1}{3} ( |Du|_{AB}^2 - |Dv|_{AB}^2 ) + \frac{1}{3} D|u|_{AB}^2 \nonumber \\
       =& D\left( \frac{1}{4} |u|^2 + \frac{1}{4} |\widetilde w|^2 + \frac{1}{3} |u|_{AB}^2 + \frac{1}{3} |Du|_{AB}^2 - \Re (u,F\,Du) \right) \nonumber \\
        & {} + |u|_L^2 + \frac{1}{4}|D^2u|^2 - 3|Du|^2.
\end{align}
Adding the time series~\eqref{eq:UParabolicProduct2} from $m = 2$ to
$n$ and using the identity $\widetilde w^n=\widetilde u^{n+1}$ we
obtain
\begin{eqnarray} \label{eq:SumParabolicProduct}
	M_1 &\geq& \frac{1}{4}|u^n|^2 + \frac{1}{4}|\widetilde{u}^{n+1}|^2 + \, \frac{1}{3} |u^n|_{AB}^2 + \frac{1}{3} |(Du)^n|_{AB}^2 + \sum_{m=2}^n |u^n|_L^2 \nonumber \\
	    & & + \frac{1}{4}\sum_{m=2}^n |(D^2u)^n|^2 - 3 \sum_{m=2}^n |(Du)^m|^2 - \Re(u^n,F\,(Du)^n)
\end{eqnarray}
where 
$$ M_1 = \frac{1}{4}|u^1|^2 + \frac{1}{4}|\widetilde{u}^{2}|^2 + \frac{1}{3} |u^1|_{AB}^2 + \frac{1}{3} |(Du)^1|_{AB}^2 - \Re(u^1,F\,(Du)^1). $$
Using Cauchy-Schwarz and Young's inequalities along with the
parabolicity relation~\eqref{eq:FtoAB} and the fact that $F$ is a
Hermitian operator, the last term $-\Re (u^n,F\,(Du)^n)$
in~\eqref{eq:SumParabolicProduct} is itself estimated as follows:
\begin{eqnarray}
	-\Re(u^n,F\,(Du)^n) &=& -\Re(F\,u^n,(Du)^n) \nonumber \\
	               &\geq & -|F\,u^n||(Du)^n| \nonumber \\
	               &\geq & -\frac{1}{12} |F\,u^n|^2 - 3|(Du)^n|^2 \nonumber \\
	               &\geq & -\frac{1}{3} |u^n|_{AB}^2 - 3|(Du)^n|^2. \nonumber
\end{eqnarray}
Equation~\eqref{eq:SumParabolicProduct} may thus be re-expressed in
the form
\begin{eqnarray}
	\frac{1}{4}|u^n|^2 &+& \frac{1}{4}|\widetilde{u}^{n+1}|^2 + \frac{1}{3} |(Du)^n|_{AB}^2 + \sum_{m=2}^n |u^n|_L^2 + \frac{1}{4}\sum_{m=2}^n |D^2u|^2 \nonumber \\
	&\leq & M_1 + 3|(Du)^n|^2 + 3 \sum_{m=2}^n |(Du)^m|^2. \label{eq:SumWithDu} 
\end{eqnarray}
Finally, applying Lemma~\ref{lem:DuParabolicBound} below to the last two
terms on the right-hand side of equation~\eqref{eq:SumWithDu} we
obtain
$$
   3|(Du)^n|^2 + 3 \sum_{m=2}^n |(Du)^m|^2 \leq 3 M_2, 
$$
where the constant $M_2$ is given by
equation~\eqref{eq:DuBoundConstant} below, and the proof of the
theorem is thus complete.
\end{proof}

The following lemma, which provides a bound on sums of squares of the
temporal difference $Du$, is used in the proof of
Theorem~\ref{thm:ParabolicStability} above.
\begin{lemma} \label{lem:DuParabolicBound} The solution $u$ of the
  Fourier-based BDF2-ADI scheme ~\eqref{eq:DiscreteParabolic} for
  equation~\eqref{eq:ParabolicEq} with initial conditions $u^0$, $u^1$
  satisfies
\begin{equation} \label{eq:DuParabolicBound}
	|(Du)^n|^2 + |u^n|^2_L + \frac{1}{2} \left( \, |(Du)^n|_A^2 + |(Du)^n|_B^2 \, \right) + \sum_{m=2}^n |(Du)^m|^2 \leq M_2
\end{equation}
for $n \geq 2$, where
\begin{equation} \label{eq:DuBoundConstant}
  M_2 = |(Du)^1|^2 + |u^1|^2_L + \frac{1}{2} \left( \, |(Du)^1|_A^2 + |(Du)^1|_B^2 \, \right).
\end{equation}
\end{lemma}

\begin{proof}

We start by taking the inner product of equation~\eqref{eq:DiscreteParabolic2}
with $Du$ to obtain
\begin{alignat}{7} \label{eq:DuParabolicProduct}
	0 &= (Du,\widehat{D}u) && + (Du,Lu) && - (Du,F\,D^2u) && + b(Du,AB\,Du) \\
	  &= (I)           && + (\mathit{II})    && + (\mathit{III})        && + (\mathit{IV}). \nonumber
\end{alignat}
We now estimate each of the terms $(I)$ through $(\mathit{IV})$ in
turn; as it will become apparent, the main challenge in this proof
lies in the estimate of the term $(\mathit{III})$.

\paragraph{$(I)$:}

Using~\eqref{eq:DuHatFormDDu} and the finite difference
product rule~\eqref{eq:ProductRule}, $(I)$ can be expressed in the form
\begin{align}
 \Re (I) &= \Re (Du,Du + \frac{1}{2}D^2u) \nonumber \\
     &= |Du|^2 + \frac{1}{4} D|Du|^2 + \frac{1}{4} |D^2u|^2. \label{eq:DuProductPart1}
\end{align}

\paragraph{$(\mathit{II})$:}

Using equation~\eqref{eq:ProductRule} we obtain
$$ \Re (\mathit{II}) = \Re (Du,Lu) = \frac{1}{2} D(u,Lu) + \frac{1}{2} (Du,L\,Du). $$
Since $L=A+B+F$ we may write
\begin{equation} \label{eq:DuProductPart2}
  \Re (\mathit{II}) = \frac{1}{2} D|u|_L^2 + \frac{1}{2} |Du|_{A+B}^2 + \frac{1}{2} (Du,F\,Du).
\end{equation}
The last term in this equation (which is a real number in view of the
Hermitian character of the operator $F$) will be used below to cancel
a corresponding term in our estimate of $(\mathit{III})$.

\paragraph{$(\mathit{III})$:}

Using~\eqref{eq:VDefinition} together with the second equation
in~\eqref{eq:DuvDef}, $(\mathit{III})$ can be expressed in the form
\begin{align} \label{eq:DuFDDu}
	(\mathit{III}) &= -(Du,F\,D^2 u) \nonumber \\
        &= -\frac{1}{2} (Du,F\,Du) + \frac{1}{2} (Du,F\,Dv) - \frac{1}{2} (Du,F\,D^2 u).
\end{align}
The first term on the right-hand side of~\eqref{eq:DuFDDu} will be
used to cancel the last term in~\eqref{eq:DuProductPart2}. Hence it
suffices to obtain bounds for the second and third terms on the
right-hand side of equation~\eqref{eq:DuFDDu}.

To estimate the second term in~\eqref{eq:DuFDDu} we consider the relation
\begin{equation}\label{eq:DuFDv}
  \frac 12 (Du,F\,Dv) = \frac 12\gamma\,\Delta t \, (Du,\delta_x\delta_y Dv) = -\frac{\gamma}{4}\,\Delta t \, (\delta_x \, Du,\delta_y \, Dv) - \frac{\gamma}{4}\,\Delta t \, (\delta_y \, Du,\delta_x \, Dv),
\end{equation}
which follows from the fact that $\delta_x$ and $\delta_y $ are
skew-Hermitian operators. Taking real parts and applying the
Cauchy-Schwarz and Young inequalities together with the parabolicity
condition~\eqref{eq:ParabolicCondition} we obtain
\begin{align}
	\frac{1}{2} \Re(Du,F\,Dv) \geq & -\frac{\sqrt{\alpha \beta}}{2} \, \Delta t \, \left( \frac{1}{2} \sqrt{\frac{\alpha}{\beta}} \, |\delta_x \, Du|^2 + \frac{1}{2} \sqrt{\frac{\beta}{\alpha}} \, |\delta_y \, Dv|^2 \right) \nonumber \\
			    & {} - \frac{\sqrt{\alpha \beta}}{2} \, \Delta t \, \left( \frac{1}{2} \sqrt{\frac{\beta}{\alpha}} \, |\delta_y \, Du|^2 + \frac{1}{2} \sqrt{\frac{\alpha}{\beta}} \, |\delta_x \, Dv|^2 \right) \nonumber \\
			   =& -\frac{1}{4} \Delta t \, ( \alpha \, |\delta_x \, Du|^2 + \beta \, |\delta_y \, Du|^2 ) - \frac{1}{4} \Delta t \, ( \alpha \, |\delta_x \, Dv|^2 + \beta \, |\delta_y \, Dv|^2 ) \nonumber \\
			   =& -\frac{1}{4} |Du|^2_{A+B} - \frac{1}{4} |Dv|^2_{A+B}. \label{eq:DuFDDuPart1}
\end{align}

To estimate third term in~\eqref{eq:DuFDDu} we once again use the
Cauchy-Schwarz and Young inequalities and we exploit the
relation~\eqref{eq:FtoAB}; we thus obtain
\begin{align}
	-\frac{1}{2} \Re(Du,F\,D^2 u) &= -\frac{1}{2} \Re(F\,Du,D^2u) \nonumber \\
	                           &\geq -\frac{1}{6} |F\,Du|^2 - \frac{3}{8} |D^2u|^2 \nonumber \\
                              &\geq -\frac{2}{3} |Du|^2_{AB} - \frac{3}{8} |D^2 u|^2. \label{eq:DuFDDuPart2}
\end{align}
Taking the real part of~\eqref{eq:DuFDDu} and using
equations~\eqref{eq:DuFDDuPart1} and~\eqref{eq:DuFDDuPart2} we obtain
the relation
\begin{equation} \label{eq:DuProductPart3}
	\Re (\mathit{III}) \geq -\frac{1}{2} \Re(Du,F\,Du) - \frac{1}{4} |Du|^2_{A+B} - \frac{1}{4} |Dv|^2_{A+B} - \frac{2}{3} |Du|^2_{AB} - \frac{3}{8} |D^2 u|^2,
\end{equation}
which, as shown below, can be combined with the estimates for $(I)$,
$\mathit{(II)}$, and $\mathit{(IV)}$ to produce an overall estimate
that consists solely of non-negative and telescoping terms---as
desired.

\paragraph{$(\mathit{IV})$:}

In view of~\eqref{eq:Seminorm} we see that $(\mathit{IV})$ coincides with
the $P$-seminorm of $Du$ with $P = AB$,
\begin{equation} \label{eq:DuProductPart4}
  \Re (\mathit{IV}) = (\mathit{IV}) = \frac 23 |Du|^2_{AB}.
\end{equation}
This term is non-negative and it therefore does not require further
treatment.

To complete the proof of the lemma we take real parts in
equation~\eqref{eq:DuParabolicProduct} and we
substitute~\eqref{eq:DuProductPart1}, \eqref{eq:DuProductPart2},
\eqref{eq:DuProductPart3} and \eqref{eq:DuProductPart4}; the result is
\begin{align}
  0 \geq& |Du|^2 + \frac{1}{4} D|Du|^2 - \frac 18 |D^2u|^2 + \frac{1}{2} D|u|_L^2 + \frac{1}{4} |Du|_{A+B}^2 - \frac{1}{4} |Dv|^2_{A+B}\nonumber \\
       =& |Du|^2 - \frac 18 |D^2u|^2 + D\left( \frac{1}{4} |Du|^2 + \frac{1}{2}|u|_L^2 + \frac{1}{4} |Du|_{A+B}^2 \right). \label{eq:DesiredBound}
\end{align}
The first two terms on the right-hand-side can be bounded by expanding
$|D^2u|^2$ and using Cauchy-Schwarz and Young's inequalities to obtain
\begin{align}
  |Du|^2 - \frac 18 |D^2u|^2 &= |Du|^2 - \frac 18 |Du - Dv|^2 \\
                             &= |Du|^2 - \frac 18 (|Du|^2 + |Dv|^2) + \frac 14 \Re (Du,Dv) \\
                             &\geq |Du|^2 - \frac 18 (|Du|^2 + |Dv|^2) - \frac 14 |Du||Dv| \\
                             &\geq |Du|^2 - \frac 14 (|Du|^2 + |Dv|^2) \\
                             &= \frac 12 |Du|^2 + \frac 14 D|Du|^2. 
\end{align}
Substituting this result into~\eqref{eq:DesiredBound}, we obtain
\begin{align}
     0 \geq& \frac 12 |Du|^2 + D\left( \frac{1}{2} |Du|^2 + \frac{1}{2}|u|_L^2 + \frac{1}{4} |Du|_{A+B}^2 \right), \label{eq:DesiredBound2}
\end{align}
which, as needed, is expressed as a sum of non-negative and telescoping
terms. Adding the time-series~\eqref{eq:DesiredBound2} from $m = 2$ to $n$
yields the desired equation~\eqref{eq:DuParabolicBound}, and the proof
is thus complete.
\end{proof}

\begin{remark} It is interesting to point out that
  Lemma~\ref{lem:DuParabolicBound} by itself implies a weak stability
  result that follows from equation~\eqref{eq:DuParabolicBound} and the Cauchy-Schwarz
  inequality:
\begin{align}
  |u^n| &= |u^1 + \sum_{m=2}^n (Du)^m| \nonumber\\
        &\leq |u^1| + \sum_{m=2}^n |(Du)^m| \nonumber\\
        &\leq |u^1| + \left( n\sum_{m=2}^n |(Du)^m|^2 \right)^{\frac{1}{2}}\nonumber \\
        &\leq |u^1| + \sqrt{nM_2},\label{eq:WeakerBound}
\end{align}
Theorem~\ref{thm:ParabolicStability} provides a much tighter energy
estimate than~\eqref{eq:WeakerBound}, of course.
\end{remark}

\subsection{Non-periodic (Legendre based) BDF2-ADI stability:
  parabolic equation\label{sec:LegendreParabolic}}

The stability result provided in the previous section for the
parabolic equation with periodic boundary conditions can easily be
extended to a non-periodic setting using a spatial approximation based
on use of Legendre polynomials. Here we provide the main
necessary elements to produce the extensions of the proofs. Background
on the polynomial collocation methods may be found
e.g. in~\cite{kopriva_implementing_2009}.

Under Legendre collocation we discretize the domain $\Omega =
[-1,1]\times[-1,1]$ by means of the $N+1$ Legendre Gauss-Lobatto
quadrature nodes $x_j=y_j$ ($j=0,\dots, N$) in each one of the
coordinate directions, which defines the grid $\{(x_j,y_k):0\leq j,k
\leq N\}$ (with $x_0=y_0=-1$ and $x_N=y_N=1$).  For {\em real-valued}
grid functions $f=\{f_{jk}\}$ and $g=\{g_{jk}\}$ we use the inner product
\begin{equation}\label{eq:LegendreInnerProduct}
  (f,g) = \sum_{j=0}^N \sum_{k=0}^N w_j w_k f_{jk} g_{jk},
\end{equation}
where $w_\ell$ ($0\leq \ell\leq N$) are the Legendre Gauss-Lobatto
quadrature weights. The interpolant $f_N$ of a grid function $f$ is a
linear combination of the form
$$
  f_N(x,y) = \sum_{j=0}^N\sum_{k=0}^N \widehat f_{jk} P_j(x)P_k(y).
$$
of Legendre polynomials $P_j$, where $\widehat f_{jk}$ are the Legendre
coefficients of $f_N$.

A certain exactness relation similar to the one we used in the Fourier
case exists in the Legendre context as well. Namely, {\em for grid
  functions $f$ and $g$ for which the product of the interpolants has
  polynomial degree $\leq 2N-1$ in the $x$ (resp. $y$) variable}, the
$j$ (resp. $k$) summation in the inner
product~\eqref{eq:LegendreInnerProduct} of the two grid functions is
equal to the integral of the product of their corresponding polynomial
interpolants with respect to $x$
(resp. $y$)~\cite[Sec.~5.2.1]{hesthaven_spectral_2007}---i.e.
\begin{subequations} \label{eq:ExactLegendre}
\begin{gather}
  (f,g) = \sum_{k=0}^N \int_{-1}^1 f_N(x,y_k)g_N(x,y_k)\,dx, \label{eq:ExactLegendreX} \\
  \text{provided} \nonumber \\
  \deg( f_N(x,y_k) g_N(x,y_k) )  \leq 2N-1 \text{ for all } 0\leq k \leq N, \nonumber \\
  \text{and} \nonumber \\
  (f,g) = \sum_{j=0}^N \int_{-1}^1 f_N(x_j,y)g_N(x_j,y)\,dy, \label{eq:ExactLegendreY} \\
  \text{provided} \nonumber \\
  \deg( f_N(x_j,y) g_N(x_j,y) )  \leq 2N-1 \text{ for all } 0\leq j \leq N, \nonumber
\end{gather}
\end{subequations}
Thus, for example, defining the Legendre $x$-derivative operator
$\delta_x$ as the derivative of the Legendre interpolant
(cf.~\eqref{eq:DerivativeExact}) (with similar definitions for
$\delta_y$, $\delta_{xx}$, $\delta_{yy}$ etc.), the exactness
relation~\eqref{eq:ExactLegendreX} holds whenever one or both of the
grid functions $f$ and $g$ is a Legendre $x$-derivative of a grid
function $h=\{h_{jk}\}$, $0\leq j,k\leq N$.

A stability proof for the parabolic equation with zero Dirichlet
boundary conditions,
$$
  U_t = \alpha\, U_{xx} + \beta\, U_{yy} + \gamma\, U_{xy} \text{ in } \Omega,\quad U=0 \text{ on } \partial \Omega,
$$
can now be obtained by reviewing and modifying slightly the strategy
presented for the periodic case in Section~\ref{sec:Parabolic}. Indeed,
the latter proof relies on the following properties of the spatial
differentiation operators: 
\begin{enumerate}

\item The discrete first and second derivative operators are skew-Hermitian and
Hermitian, respectively.

\item The operators $A$, $B$, $L$ and $AB$ defined in
  Section~\ref{sec:Parabolic} are positive semi-definite.
\end{enumerate}
Both of these results were established using the exactness relation
between the discrete and integral inner products together with
vanishing boundary terms arising from integration by parts---which
also hold in the present case since the exactness
relations~\eqref{eq:ExactLegendre} are only required in the proof to
convert inner products involving derivatives, so that the degree of
polynomial interpolants will satisfy the requirements of the
relations~\eqref{eq:ExactLegendre}. Since all other aspects of the
proofs in Section~\ref{sec:Parabolic} are independent of the
particular spatial discretization or boundary conditions used, we have
the following theorem:

\begin{theorem} 
  The stability estimate given in Theorem~\ref{thm:ParabolicStability}
  also holds on the domain $[-1,1]\times[-1,1]$ with homogeneous
  boundary conditions using the Legendre Gauss-Lobatto collocation
  method, where the inner products and norms are taken to be the
  Legendre versions.
\end{theorem}

\begin{remark}
  Unfortunately, the stability proof cannot be extended directly to
  the non-periodic hyperbolic case. Indeed, since in this case only
  one boundary condition is specified in each spatial direction, not
  all boundary terms arising from integration by parts vanish---and,
  hence, the first derivative operators are generally not
  skew-Hermitian. See~\cite{gottlieb_spectral_2001} for a discussion
  of spectral method stability proofs for hyperbolic problems.
\end{remark}

\section{Quasi-unconditional stability for higher-order non-ADI BDF 
  methods: periodic advection-diffusion equation\label{sec:HighOrderBDFStability}}

\subsection{Rectangular window of stability\label{sec:RectangularWindow}}

This paper does not present stability proofs for the BDF-ADI methods
of order higher than two. In order to provide some additional insights
into the stability properties arising from the BDF strategy in the
context of time-domain PDE solvers, this section investigates the
stability of the BDF schemes of order $s\geq
2$---cf. Remark~\ref{rem:vonNeumann} as well as the last paragraph in
Section~\ref{sec:StabilityDiscussion}---under periodic boundary
conditions and Fourier discretizations. Because of Dahlquist's second
barrier~\cite[p. 243]{lambert_numerical_1991} the $s\geq 3$ schemes
cannot be unconditionally stable for general (even linear)
PDEs. However, we will rigorously establish that the BDF methods of
order $s$ with $2\leq s \leq 6$ are \emph{quasi-unconditionally
  stable} for the advection-diffusion equation---in the sense of
Definition~\ref{def:QuasiUnconditional}. (As shown in
Section~\ref{sec:BDF2Stability} further, the $s=2$ algorithms are
indeed unconditionally stable, at least for certain linear PDE.)

To introduce the main ideas in our quasi-unconditional stability
analysis for BDF-based schemes we consider first a Fourier-BDF
scheme for the advection-diffusion equation in one spatial dimension
with periodic boundary conditions:
\begin{align} \label{eq:AdvectionDiffusion}
	& U_t + \alpha U_x = \beta U_{xx}, \quad x \in \mathbb{R},\quad t \geq 0, \\
	& U(x,0) = f(x), \quad U(x,t) = U(x+2\pi,t), \nonumber
\end{align}
where $\beta > 0$. Using the $N$-point Fourier discretization
described in Sections~\ref{sec:PrelimDefs} and~\ref{sec:DiscreteOperators}, the
resulting semi-discrete equation is given by
\begin{equation} \label{eq:DiscreteAdvectionDiffusion}
  \frac{\partial}{\partial t} u = (-\alpha\,\delta_x +
  \beta\,\delta_x^2)\,u.
\end{equation}
As mentioned in Remark~\ref{rem:vonNeumann}, the von Neumann criterion
provides a necessary and sufficient stability condition for this
problem: the order-$s$ scheme is stable if and only if the (complex!)
eigenvalues of the spatial operator in the semi-discrete
system~\eqref{eq:DiscreteAdvectionDiffusion} multiplied by $\Delta t$
lie within the region $R_s$ of absolute stability of the BDF method of
order $s$. We will see that, for the present one-dimensional
advection-diffusion problem, these eigenvalues lie on a parabola
which does not vary with $N$. As is known~\cite[Sec. 7.6.1 and
p. 174]{leveque_finite_2007}, further, the boundary of $R_s$ is given
by the polar parametrization
\begin{equation}\label{eq:BDFStabBoundary}
  z(\theta) =  \frac{1}{b} \left( 1 - \sum_{j=0}^{s-1} a_j e^{-i(j+1)\theta}  \right).
\end{equation}
To establish the quasi-unconditional stability
(Definition~\ref{def:QuasiUnconditional}) of the Fourier-BDF scheme
under consideration it is therefore necessary and sufficient to show
that a certain family of ``complete parabolas'' lie in the stability
region of the BDF scheme for $\Delta t < M_t$ and $\Delta x < M_h$ for
some constants $M_t$ and $M_h$ which define the corresponding
rectangular window of stability. As discussed in the next section,
further, certain CFL-like stability constraints that hold outside of
the rectangular window of stability are obtained by consideration of
the relative position of eigenvalues on such parabolas and the BDF
stability region. The former property (quasi-unconditional stability)
follows from an application of Lemma~\ref{lem:Parabola}, which
establishes that the stability regions of the BDF schemes contain the
required families of parabolas.

\begin{figure}[!htb]
	\centering
	\includegraphics[width=0.5\textwidth]{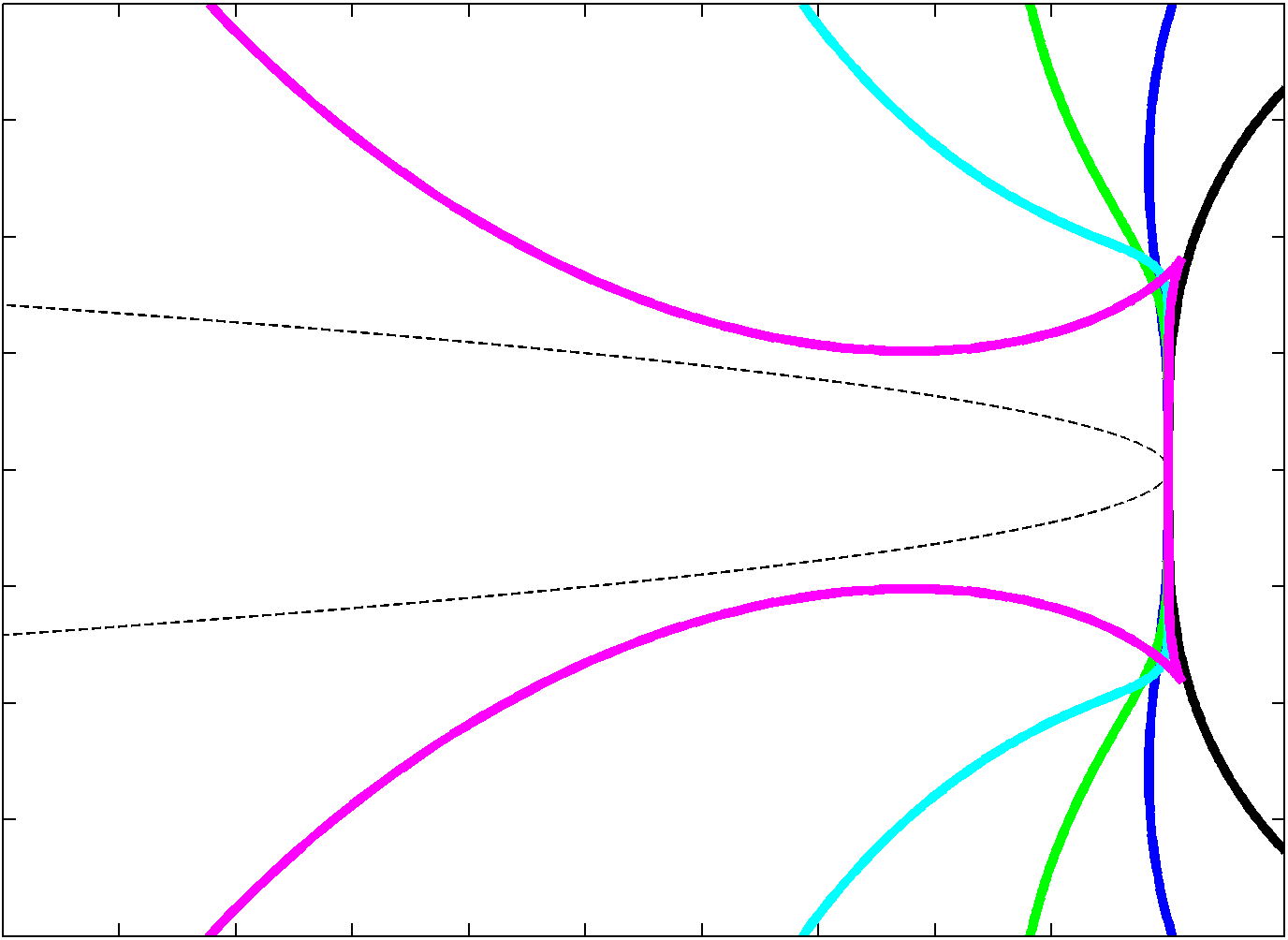}
        \caption{Boundaries of the stability regions of the BDF
          methods of order $s=2$ (black), 3 (blue), 4 (green), 5
          (cyan) and 6 (magenta), together with a parabola $\Gamma_m$
          (dashed black), $m=0.1$, that is contained in every
          stability region.}
	\label{fig:Parabola}
\end{figure}

\begin{lemma} \label{lem:Parabola} Let $\Gamma_m$ denote the locus of
  the left-facing parabola of equation $x = -\frac{1}{m} y^2$
  ($m>0$). Then, for each $2\leq s \leq 6$, there exists a critical
  $m$-value $m_C$ such that $\Gamma_m$ lies in the stability region
  $R_s$ of the BDF method of order $s$ for all $0\leq m<m_C$.
\end{lemma}
\begin{proof}
  Clearly, given $s=1,\dots,6$, $\Gamma_m$ is contained within $R_s$
  (see Figure~\ref{fig:Parabola}) if and only if $m$ is larger than
  the infimum
\begin{equation} \label{eq:mC} 
m_C = \inf \left\lbrace m \,|\, m =
    -\frac{y^2}{x}\, \mbox{for some}\, (x+iy) \in R_s\right\rbrace.
\end{equation}
Simple geometrical considerations, further, show that $m_C$ is a
positive number (details in these regards can be found
in~\cite[Sec.~2.4]{cubillos_general-domain_2015}; consideration
of Figure~\ref{fig:Parabola}, however, clearly indicates that $m_C>0$). Thus, the
parabola $\Gamma_m$ is contained within $R_s$ if and only if $0\leq
m<m_C$, and the proof is complete.
\end{proof}

Numerical values of $m_C$ for each BDF method of orders 3 through 6
(which were computed as the infimum of $-\frac{y^2}{x}$ over the
boundary of $R_s$), are presented in Table~\ref{table:ConstantM_C}.
\begin{table}[!htb]
\centering
\begin{tabular}{c|cccc}
\hline \hline
 $s$ & 3 & 4 & 5 & 6 \\ \hline
 $m_C$ & 14.0 & 5.12 & 1.93 & 0.191 \\ \hline \hline
\end{tabular}
\caption{Numerical values of the constant $m_C$ obtained via  numerical evaluation of equation~\eqref{eq:mC}.  For all $m<m_C$
  the parabola $\Gamma_{m}$ described in Lemma~\ref{lem:Parabola} is contained in
  the region of absolute stability of the BDF method of order $s$. By
  Theorem~\ref{thm:AdvectionDiffusion1D}, the order-$s$ BDF method applied to the
  advection-diffusion equation $u_t + \alpha \, u_x = \beta \,u_{xx}$ with
  Fourier collocation is stable for all $\Delta t < \frac{\beta}{\alpha^2}m_C$.}
\label{table:ConstantM_C}
\end{table}

\begin{theorem} \label{thm:AdvectionDiffusion1D} Let $2 \leq s \leq
  6$. The solution produced for the
  problem~\eqref{eq:AdvectionDiffusion} by the $s$-th order
  Fourier-based BDF scheme described in this section is
  quasi-unconditionally stable
  (Definition~\ref{def:QuasiUnconditional}) with $M_h = \infty$ and
  $M_t = \frac{\beta}{\alpha^2}m_C$ ($M_t=\infty$ for $\alpha = 0$),
  where the $s$ dependent constant $m_C$ is given in
  equation~\eqref{eq:mC}.
\end{theorem}
\begin{proof}
Applying the discrete Fourier transform, 
$$ \widehat{u}_k = \frac{1}{N+1} \sum_{j=0}^N u_j e^{-i x_j k}, \quad -\frac{N}{2} \leq k \leq \frac{N}{2} $$
to equation~\eqref{eq:DiscreteAdvectionDiffusion} we obtain the set of ODEs
\begin{equation} \label{eq:FourierAdvectionDiffusion}
	\frac{\partial}{\partial t} \widehat{u}_k = -(i\alpha k + \beta k^2)\,\widehat{u}_k
\end{equation}
for the Fourier coefficients $\widehat{u}_k$. It is clear from this transformed
equation that the eigenvalues of the spatial operator for the semi-discrete
system are given by
\begin{equation} \label{eq:EigenvaluesAdvectionDiffusion}
	\lambda(k) = -(i\alpha k + \beta k^2).
\end{equation}
In view of Remark~\ref{rem:vonNeumann}, to complete the proof it
suffices to show that these eigenvalues multiplied by $\Delta t$ lie
in the stability region of the BDF method for all $\Delta t <
\frac{\beta}{\alpha^2}m_C$.

Let $z_k = \lambda(k) \, \Delta t$. If $\alpha=0$, then $z_k$ is a
non-positive real number for all integers $k$. In view of the
A(0)-stability of the BDF methods, we immediately see that the methods
are unconditionally stable in this case. For the case $\alpha \neq 0$,
in turn, we have
\begin{align}
  z_k &= \lambda(k)\Delta t = - \beta\,\Delta t\,k^2-i\alpha\,\Delta t\,k  \nonumber\\
  &= -\frac{\beta}{\alpha^2 \Delta t}(\alpha\,\Delta t\,k)^2 - i(\alpha\,\Delta t\,k) \nonumber\\
  &= -\frac{1}{\frac{\alpha^2 \Delta t}{\beta}} (\alpha\,\Delta
  t\,k)^2 - i(\alpha\,\Delta t\,k)\quad \mbox{with}\quad -\frac{N}{2}
  \leq k \leq \frac{N}{2}. \label{eq:LambdaDt}
\end{align}
From~\eqref{eq:LambdaDt} it is clear that, for all integers $k$, $z_k$
lies on the  left-facing parabola $\Gamma_m $ with 
\begin{equation}\label{eq:MDefAdvDiff}
m =
\frac{\alpha^2\Delta t}{\beta}.
\end{equation}
But, by Lemma~\ref{lem:Parabola} we know that the parabola $\Gamma_m $
lies within the stability region $R_s$ for all $m<m_C$, and, thus, for
all
$$
 \Delta t < \frac{\beta}{\alpha^2}m_C.
$$
Furthermore, the above condition holds for all $N$ and all $k$ with
$-N/2\leq k\leq N/2$. It follows that $M_h=\infty$ and the proof is
complete.
\end{proof}

We now establish the quasi-unconditional stability of the
Fourier-based BDF methods for the advection-diffusion equation
\begin{equation} \label{eq:AdvectionDiffusion2D}
  u_t + \boldsymbol{\alpha} \cdot \nabla u = \beta \Delta u, \quad \text{ in } [0,2\pi]^d,\ d=2,3
\end{equation}
in two- and three-dimensional space and with periodic boundary
conditions, where $\boldsymbol{\alpha}=(\alpha_1,\alpha_2)^{\mathrm
  T}$ and $\boldsymbol{\alpha}=(\alpha_1,\alpha_2,\alpha_3)^{\mathrm
  T}$ for $d=2$ and 3 respectively.  Thus, letting $\mathbf x = (x,y)^{\mathrm
T}$, $\mathbf k = (k_x,k_y)^{\mathrm T}$, $\mathbf N = (N_x,N_y)^{\mathrm T}$
(resp. $\mathbf x = (x,y,z)^{\mathrm T}$, $\mathbf k = (k_x,k_y,k_z)^{\mathrm T}$,
$\mathbf N = (N_x,N_y,N_z)^{\mathrm T}$) in $d=2$ (resp. $d=3$) spatial dimensions,
and substituting the Fourier series (using multi-index notation)
$$
  u(\mathbf x) = \sum_{\mathbf k = -\mathbf N/2}^{\mathbf N/2} \widehat u_{\mathbf k} e^{i(\mathbf k \cdot \mathbf x)}
$$
into equation~\eqref{eq:AdvectionDiffusion2D}, the Fourier-based BDF method of
order $s$ results as the $s$-order BDF method applied to the ODE system
\begin{equation}\label{eq:AdvDiffFourierSemiDiscrete} 
  \frac{\partial \widehat u_{\mathbf k}}{\partial t} = \left(-i(\boldsymbol\alpha \cdot \mathbf k) -\beta |\mathbf k|^2 \right) \widehat u_{\mathbf k}.
\end{equation}
for the Fourier coefficients $\widehat u_{\mathbf k}$. (In order to
utilize a single mesh-size parameter $h$ and corresponding
quasi-unconditional stability constant $M_h$ while allowing for
different grid-fineness in the $x$, $y$ and $z$ directions, we utilize
positive integers $r_2$ and $r_3$ and discretize the domain on the
basis of $N_x+1$ points in the $x$ direction, $N_y+1 = r_2N_x+1$
points in the $y$ direction, and $N_z+1 = r_3N_x+1$ points in the $z$
direction ($N_x$ even). The mesh size parameter is then given by
$h=2\pi/(N_x+1)$.)

\begin{theorem} \label{thm:AdvectionDiffusion2D} The Fourier-based BDF
  scheme of order $s$ (not ADI!)  for the
  problem~\eqref{eq:AdvectionDiffusion2D} with $3 \leq s \leq 6$ is
  quasi-unconditionally stable with constants $M_t =
  \frac{|\boldsymbol{\alpha}|^2}{\beta}m_C$ and $M_h=\infty$.
\end{theorem}

\begin{proof}
We first note that the eigenvalues of the discrete spatial operator in
equation~\eqref{eq:AdvDiffFourierSemiDiscrete} multiplied by $\Delta t$ are given by
\begin{equation} \label{eq:Eigenvalues2D}
  z_{\mathbf k} = -i\Delta t\,\boldsymbol\alpha \cdot \mathbf k - \Delta t\,\beta|\mathbf k|^2.
\end{equation}
Clearly, in contrast with the situation encountered in the context of
the one dimensional problem considered in
Theorem~\ref{thm:AdvectionDiffusion1D}, in the present case the set of
$z_{\mathbf k}$ does not lie on a single parabola. But, to establish
quasi-unconditional stability it suffices to verify that this set is
bounded on the right by a certain left-facing parabola through the
origin. This can be accomplished easily: in view of the Cauchy-Schwarz
inequality, we have
\begin{equation} \label{eq:ImaginaryBound}
  |\boldsymbol\alpha \cdot \mathbf k| \leq |\boldsymbol\alpha||\mathbf k|.
\end{equation}
Therefore, letting $\xi = |\mathbf k|$, the eigenvalues multiplied by $\Delta
t$ are confined to the region
$$ 
  \{z\,:\,\Re z = -\Delta t \beta \xi^2,\; |\Im z| \leq \Delta t |\boldsymbol\alpha|\xi,\; \xi\geq 0\}.
$$
Clearly, the boundary of this region is the left-facing parabola
$$
  x = -\frac{\beta}{\Delta t|\boldsymbol\alpha|^2} y^2
$$
and the theorem now follows from an application of
Lemma~\ref{lem:Parabola} together with a simple argument similar to
the one used in Theorem~\ref{thm:AdvectionDiffusion1D}.
\end{proof}

\subsection{Order-$s$ BDF methods outside the rectangular window of stability\label{sec:StabilityBDFvsAB}}

\begin{figure}[!htb]
\centering
$\begin{array}{ccc}
	  \includegraphics[width=0.25\textwidth]{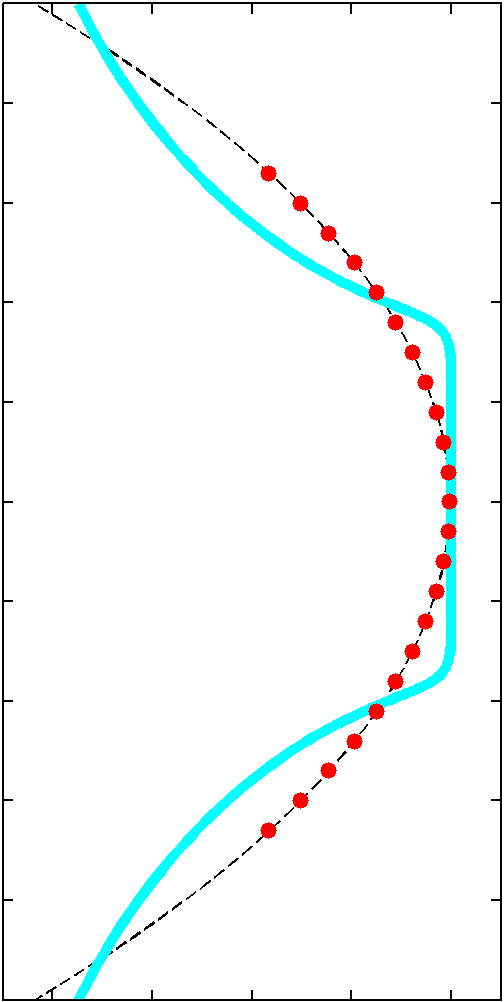} &
    \includegraphics[width=0.25\textwidth]{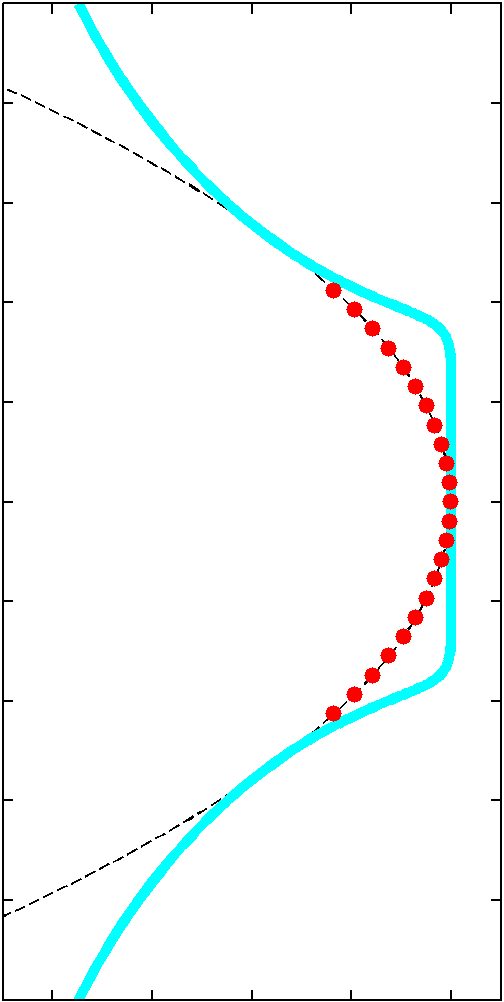} &
    \includegraphics[width=0.25\textwidth]{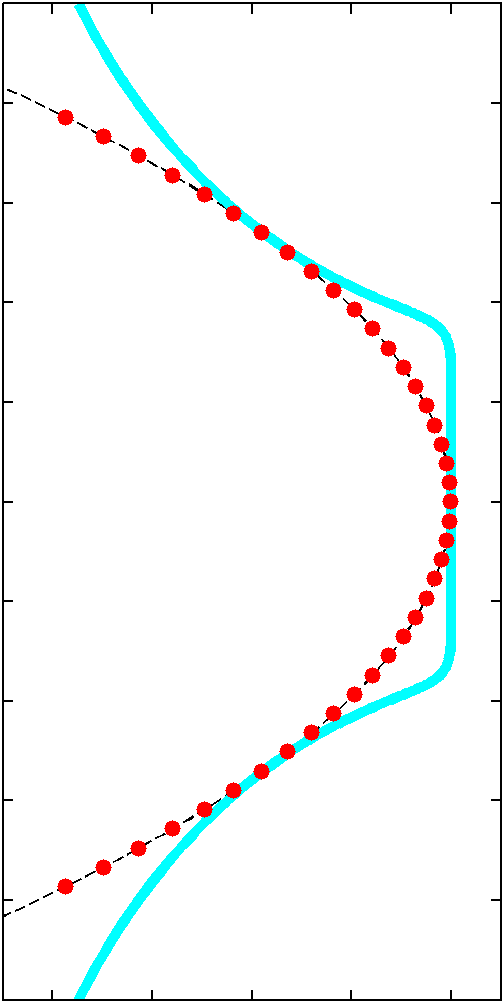} \\
    \text{(a)} & \text{(b)} & \text{(c)} 
\end{array}$
\caption{Variation of the eigenvalue distribution for the
  one-dimensional advection-diffusion equation with $\alpha=1.0$ and
  $\beta = 0.05$ (parameters selected for clarity of visualization),
  as $N$ and $\Delta t$ are varied. (Theoretical value: $M_t=0.0965$
  for this selection of physical parameters.) The eigenvalues
  associated with this problem multiplied by $\Delta t$ are plotted as
  red dots together with the corresponding parabola $\Gamma_m$ with $m
  = \alpha^2 \Delta t / \beta$ (dashed black line) and the boundary of
  the BDF5 stability region (cyan curve; cf.
  Figure~\ref{fig:Parabola}) for various values of $N$ and $\Delta
  t$. Figure~\ref{fig:AdvDiffCFL}(a): $N+1=19$ and $\Delta t = 0.15$;
  in this case some eigenvalues lie outside the stability region.
  Figure~\ref{fig:AdvDiffCFL}(b): The time-step is reduced to $\Delta
  t = M_t = 0.0965$, the parabola $\Gamma_m$ is tangent to the
  stability boundary and fully contained in stability region; in
  particular, all eigenvalues now lie within in the stability region.
  Figure~\ref{fig:AdvDiffCFL}(c) The number of grid points is
  increased to $N+1=41$ while maintaining stability: for this value of
  $\Delta t$ stability holds for all values of $N$.}
\label{fig:AdvDiffCFL}
\end{figure}

Theorems~\ref{thm:AdvectionDiffusion1D}
and~\ref{thm:AdvectionDiffusion2D} should not be viewed as a
suggestion that the $s$-th order BDF methods are not stable when the
constraint $\Delta t < M_t$ in the theorem is not satisfied. For
example, for $\Delta t > M_t$ the complete parabolas $\Gamma_m$
defined in Section~\ref{sec:RectangularWindow} intersect the region where
the BDF method is unstable---as demonstrated in
Figure~\ref{fig:AdvDiffCFL}. Fortunately, however, stability can
still be ensured for such values of $\Delta t$ provided sufficiently
large values of the spatial meshsizes are used.  For example, in the
one-dimensional case considered in
Theorem~\ref{thm:AdvectionDiffusion1D} we have $N + 1 = 2\pi/h$, and
the eigenvalues are given by equation~\eqref{eq:LambdaDt}: clearly
only a bounded segment in the parabola is actually relevant to the
stability of the ODE system that results for each fixed value of
$N$. In particular, we see that stability is ensured provided this
particular segment, and not necessarily the complete parabola
$\Gamma_m$, is contained in the stability region of the $s$-th order
BDF algorithm.

From equation~\eqref{eq:LambdaDt} we see that increases in the values
of $N$ lead to corresponding increases in the length of the parabolic
segment on which the eigenvalues actually lie, while decreasing
$\Delta t$ results in reductions of both the length of the relevant
parabolic segment as well as the width of the parabola
itself. Therefore, for $\Delta t > M_t$, increasing the number of grid
points eventually causes some eigenvalues to enter the region of
instability. But stability can be restored by a corresponding
reduction in $\Delta t$---see Figure~\ref{fig:AdvDiffCFL}. In other
words, a CFL-like condition of the form $\Delta t \leq F(h)$
($h=2\pi/(N+1)$) exists for $\Delta t > M_t$: the ``maximum stable
$\Delta t$'' function $F(h)$ can be obtained by considering the
intersection of the boundary locus of the BDF stability region
(equation~\eqref{eq:BDFStabBoundary}) and the parabola $\Gamma_m$ with
$m$ given by equation~\eqref{eq:MDefAdvDiff}. It can be seen from the first
line in equation~\eqref{eq:LambdaDt} that, provided the coefficient of
$\Delta t$ in the real part is much smaller than the corresponding
coefficient of $\Delta t$ in the imaginary part then the CFL-like
condition will be approximately linear around that point---as is
apparent by consideration of the actual curves $\Delta t_\mathrm{max}
= F(h)$ in Figure~\ref{fig:StabBDFvsAB} near $h=1$.  Of course, when
$\Delta t$ is reduced to the value $M_t$ or below, then no increases
in $N$ (reductions in $h$) result in instability---as may be
appreciated by consideration of Figures~\ref{fig:AdvDiffCFL}. We may thus
emphasize: within the rectangular stability window no such CFL-like stability
constraints exist.

\begin{figure}[!htb]
	\centering
	\includegraphics[width=0.7\textwidth]{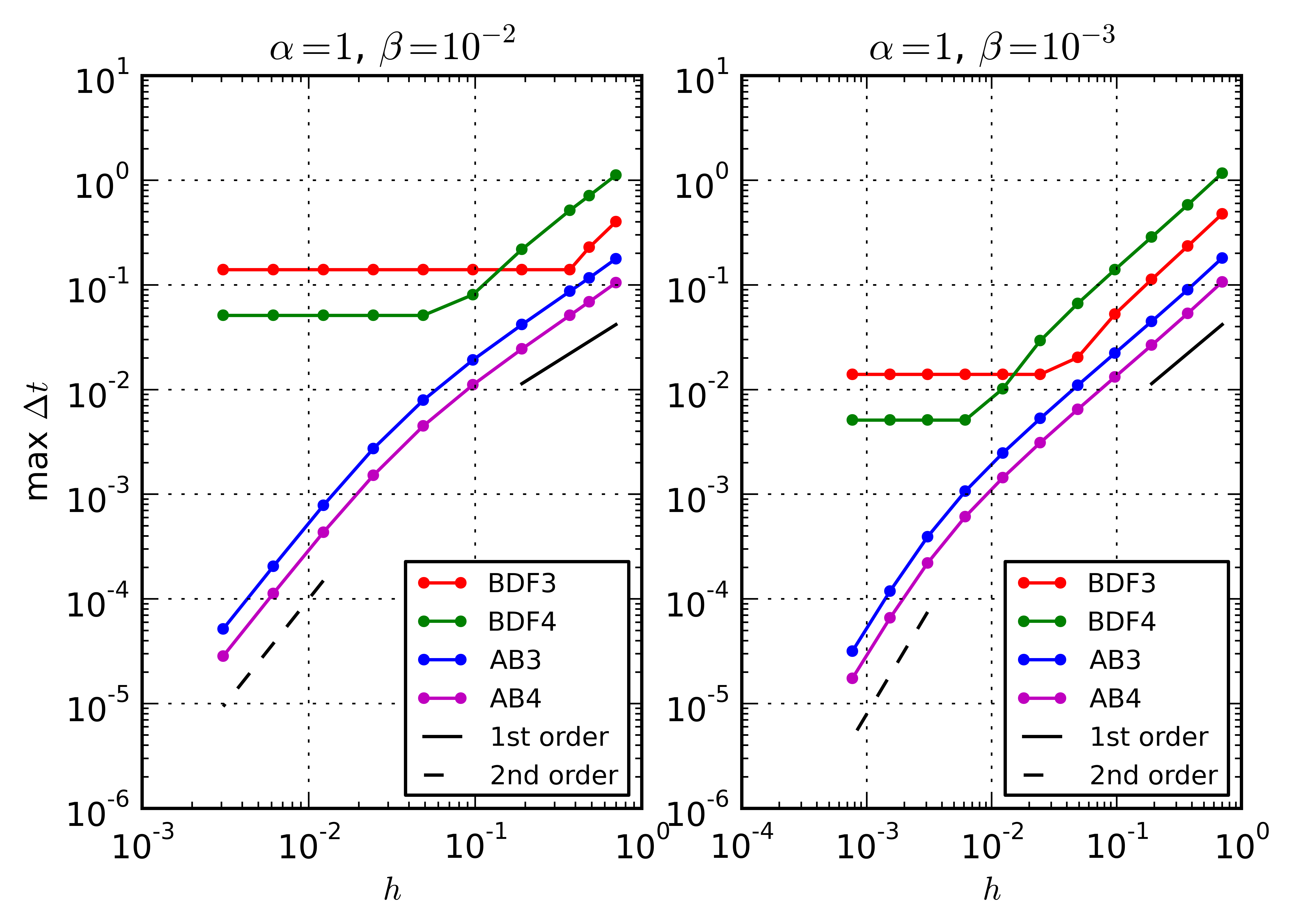}
        \caption{Maximum stable $\Delta t$ versus spatial mesh size
          $h$ for Fourier-based BDF and AB methods of orders three and
          four for the advection-diffusion
          equation~\eqref{eq:AdvectionDiffusion}, with $\alpha=1$,
          $\beta=10^{-2}$ (left plot) and $\alpha=1$, $\beta=10^{-3}$
          (right plot). Both the BDF and AB maximum-stable-$\Delta t$
          values were obtained by considering the intersection of the
          boundary locus of the relevant stability region (either BDF
          or AB), as indicated in Section~\ref{sec:StabilityBDFvsAB}
          in the context of the BDF method.}
	\label{fig:StabBDFvsAB}
\end{figure}

For comparison, Figure~\ref{fig:StabBDFvsAB} also displays the maximum
stable $\Delta t$ curves for the Fourier-based Adams-Bashforth (AB)
multistep methods of orders three and four as functions of the
meshsize $h$ for the advection-diffusion equation under consideration
with $\alpha =1$ and two values of $\beta$. We see that the stability
of both the BDF and AB methods is controlled by an approximately
linear CFL-type constraint of the form $\Delta t < Ch$ for
sufficiently large values of $h$. For smaller values of $h$ the CFL
condition for the explicit method becomes more severe, and eventually
reaches the approximately quadratic regime $\Delta t < Ch^2$. By this
point, the BDF methods have already entered the window of
quasi-unconditional stability. For the particular value of $\alpha$
considered in these examples, at $h=\beta$ the maximum stable $\Delta
t$ values for the BDF methods are approximately one hundred times
larger than their AB counterparts. Clearly, the BDF methods are
preferable in regimes where the AB methods suffer from the severe
$\Delta t < Ch^2$ CFL condition.

\section{Quasi-unconditional stability for the full
Navier-Stokes equations: a numerical study\label{sec:BDFLinearNS}}

Tables~\ref{table:BDFStabilityRe50},
\ref{table:BDFStabilityRe100} and
\ref{table:BDFStabilityRe200} display numerically estimated
maximum stable $\Delta t$ values for the Chebyshev-based BDF-ADI
algorithm introduced in Part~I for the full Navier-Stokes equations in
two dimensional space for various numbers of Chebyshev discretization
points. The specific problem under consideration is posed in the unit
square $[0,1]\times[0,1]$ with Mach number 0.9 and various Reynolds
numbers, with initial condition given by $\mathbf{u}=0$, $\rho=T=1$,
and with a source term of the form
$$
  f(x,y,t) = A\,\sin(2\pi t)\,\exp \left(\,-\frac{1}{2\sigma^2}\left(\,(x-x_0)^2 + (y-y_0)^2\,\right)\,\right)
$$
as the right-hand side of the $x$-coordinate of the momentum equation
($A=6.0$, $\sigma^2=0.05$ and $x_0=y_0=0.5$). No-slip isothermal
boundary conditions ($\mathbf{u}=0$, $T=1$) are assumed at $y=0$ and
$y=1$, and a sponge layer (see Part~I) of thickness $0.1$ and
amplitude $2.0$ is enforced at $x=0$ and $x=1$.  The algorithm was
determined to be stable for a given $\Delta t$ if the solution does
not blow up for 20000 time steps or for the number of time steps
required to exceed $t=100$, whichever is greater.

Tables~\ref{table:BDFStabilityRe50},
\ref{table:BDFStabilityRe100} and
\ref{table:BDFStabilityRe200} suggest that the BDF-ADI
Navier-Stokes algorithm introduced in Part~I is indeed
quasi-unconditionally stable.  In particular, consideration of the
tabulated values indicates that the BDF-ADI methods may be
particularly advantageous whenever the time-steps required for
stability in a competing explicit scheme for a given spatial
discretization is much smaller than the time-step required for
adequate resolution of the time variation of the solution.
\begin{table}[!htb]
\centering
\begin{tabular}{c|ccccc} 
 \hline \hline
 & \multicolumn{5}{c}{$s = \dots$} \\
$N_y$  & 2      & 3    & 4   & 5   & 6  \\ \hline
 12    & 6.1e-1 & 3.5e-1 & 9.1e-2 & 4.1e-2 & 1.7e-2  \\
 16    & 6.1e-1 & 2.9e-1 & 8.7e-2 & 3.2e-2 & 9.0e-3  \\
 24    & 6.1e-1 & 1.3e-1 & 5.9e-2 & 1.9e-2 & 5.3e-3  \\
 32    & 6.1e-1 & 1.2e-1 & 5.0e-2 & 1.5e-2 & 4.3e-3  \\
 48    & 6.1e-1 & 1.0e-1 & 4.2e-2 & 1.3e-2 & 3.7e-3  \\ 
 64    & 6.1e-1 & 1.0e-1 & 4.1e-2 & 1.2e-2 & 3.5e-3  \\ 
 96    & 6.1e-1 & 1.0e-1 & 4.0e-2 & 1.2e-2 & 3.1e-3  \\ 
 128   & 6.1e-1 & 1.0e-1 & 4.0e-2 & 1.2e-2 & 2.8e-3  \\ \hline \hline
\end{tabular}
\caption{Maximum stable $\Delta t$ values for the order-$s$ BDF-ADI Navier-Stokes solvers introduced in Part~I with $s =
  2,\dots,6$, in two spatial dimensions, and  at Reynolds number $\mathrm{Re}=50$ and Mach number $0.8$, with various numbers  $N_y$ of
  discretization points in the $y$ variable. The number of discretization points in the $x$
  direction is fixed at $N_x = 12$.}
\label{table:BDFStabilityRe50}
\end{table}

\begin{table}[!htb]
\centering
\begin{tabular}{c|ccccc} 
 \hline \hline
 & \multicolumn{5}{c}{$s = \dots$} \\
$N_y$  & 2      & 3    & 4   & 5   & 6  \\ \hline
 12    & 6.4e-1 & 3.4e-1 & 5.9e-2 & 3.4e-2 & 1.5e-2  \\
 16    & 6.3e-1 & 2.7e-1 & 5.0e-2 & 2.4e-2 & 9.9e-3  \\
 24    & 6.3e-1 & 1.1e-1 & 4.5e-2 & 1.9e-2 & 6.1e-3  \\
 32    & 6.3e-1 & 9.2e-2 & 3.7e-2 & 1.7e-2 & 5.1e-3  \\
 48    & 6.3e-1 & 7.8e-2 & 3.2e-2 & 1.6e-2 & 4.6e-3  \\ 
 64    & 6.3e-1 & 7.4e-2 & 3.1e-2 & 1.5e-2 & 4.4e-3  \\ 
 96    & 6.3e-1 & 7.2e-2 & 3.0e-2 & 1.5e-2 & 4.3e-3  \\ 
 128   & 6.3e-1 & 7.1e-2 & 3.0e-2 & 1.5e-2 & 4.1e-3  \\ \hline \hline
\end{tabular}
\caption{Same as Table~\ref{table:BDFStabilityRe100} but with
Reynolds number $\mathrm{Re}=100$.}
\label{table:BDFStabilityRe100}
\end{table}

\begin{table}[!htb]
\centering
\begin{tabular}{c|ccccc} 
 \hline \hline
 & \multicolumn{5}{c}{$s = \dots$} \\
$N_y$  & 2      & 3    & 4   & 5   & 6  \\ \hline
 12    & 5.5e-1 & 2.9e-1 & 4.5e-2 & 2.8e-2 & 1.3e-2  \\
 16    & 5.3e-1 & 2.9e-1 & 4.4e-2 & 2.0e-2 & 8.8e-3  \\
 24    & 5.5e-1 & 1.1e-1 & 2.5e-2 & 1.3e-2 & 4.6e-3  \\
 32    & 5.4e-1 & 8.6e-2 & 2.3e-2 & 1.1e-2 & 3.6e-3  \\
 48    & 5.3e-1 & 6.6e-2 & 2.1e-2 & 9.5e-3 & 2.9e-3  \\ 
 64    & 5.3e-1 & 6.1e-2 & 2.1e-2 & 8.2e-3 & 2.9e-3  \\ 
 96    & 5.3e-1 & 5.9e-2 & 2.1e-2 & 8.3e-3 & 2.4e-3  \\ 
 128   & 5.3e-1 & 5.8e-2 & 2.1e-2 & 8.2e-3 & 2.5e-3  \\ \hline \hline
\end{tabular}
\caption{Same as Table~\ref{table:BDFStabilityRe50} but with
Reynolds number $\mathrm{Re}=200$.}
\label{table:BDFStabilityRe200}
\end{table}

\section{Summary and conclusions\label{sec:Conclusions}}

A variety of studies were put forth in this paper concerning the
stability properties of the compressible Navier-Stokes BDF-ADI
algorithms introduced in Part~I, including rigorous stability proofs
for associated BDF- and BDF-ADI-based algorithms for related linear
equations, and numerical stability studies for the fully nonlinear
problem. In particular, the present paper presents proofs of
unconditional stability or quasi-unconditional stability for BDF-ADI
schemes as well as certain associated un-split BDF schemes, for a
variety of diffusion and advection-diffusion linear equations in one,
two and three dimensions, and for schemes of orders $2\leq s\leq 6$ of
temporal accuracy. (The very concept of quasi-unconditional stability
was introduced in Part~I to describe the observed stability character
of the Navier-Stokes BDF-ADI algorithms introduced in that paper.) A
set of numerical experiments presented in this paper for the
compressible Navier-Stokes equation suggests that the algorithms
introduced in Part~I do enjoy the claimed property of
quasi-unconditional stability.

\paragraph{Acknowledgments} The authors gratefully acknowledge support
from the Air Force Office of Scientific Research and the National Science
Foundation. MC also thanks the National Physical Science Consortium for
their support of this effort.

\bibliography{ThesisBib}

\end{document}